\documentclass[12pt]{article}

\usepackage{amssymb,amsmath,amsthm,sectsty,url,bbm, units, mathtools}
\usepackage{braket, comment, color}
\usepackage[letterpaper,hmargin=1.0in,vmargin=1.0in]{geometry}
\usepackage[margin=20pt,font=small,labelfont=bf]{caption}
\usepackage[boxed]{algorithm} \floatname{algorithm}{Figure}
\usepackage{nag,fixltx2e,booktabs,flafter,hyperref}
\usepackage[T1]{fontenc}

\sectionfont{\large}
\subsectionfont{\normalsize}
\numberwithin{equation}{section}
\newtheorem{theorem}{Theorem}[section]
\newtheorem{lemma}[theorem]{Lemma}

\newtheorem{corollary}[theorem]{Corollary}

\newtheorem{definition}[theorem]{Definition}

\newtheorem{claim}[theorem]{Claim}

\DeclareMathOperator{\poly}{poly}
\DeclareMathOperator{\polylog}{polylog}

\newcommand{\N}{\mathbb N} \newcommand{\R}{\mathbb R}
\newcommand{\Z}{\mathbb Z} 
\newcommand{\F}{\mathbb F} 
\newcommand{\B}{\{ 0,1 \}} \newcommand{\BM}{\{ -1,1 \}}
\renewcommand{\H}{\mathbb H}

 \newcommand{\AM}{{\mathcal {AM}}}
\newcommand{\MA}{{\mathcal {MA}}}

\newcommand{\GHD}{{\mathsf {GHD}}} \newcommand{\ORT}{{\mathsf {ORT}}}
\newcommand{\ECC}{{\mathsf {ECC}}}

\begin{document}
\title{Arthur-Merlin Streaming Complexity} \author{Tom Gur \thanks{
    Department of Computer Science and Applied Mathematics, Weizmann
    Institute of Science, Rehovot 76100, Israel. E-mail: {\tt
      \{tom.gur, ran.raz\}@weizmann.ac.il}. Research supported by an
    Israel Science Foundation grant and by the I-CORE Program of the
    Planning and Budgeting Committee and the Israel Science
    Foundation.}  \and Ran Raz $^*$} \date{\today}
\maketitle
\begin{abstract}
  We study the power of Arthur-Merlin probabilistic proof systems in
  the data stream model. We show a canonical $\AM$ streaming algorithm
  for a wide class of data stream problems. The algorithm offers a
  tradeoff between the length of the proof and the space complexity
  that is needed to verify it.

  As an application, we give an $\AM$ streaming algorithm for the
  \emph{Distinct Elements} problem. Given a data stream of length $m$
  over alphabet of size $n$, the algorithm uses $\tilde O(s)$ space
  and a proof of size $\tilde O(w)$, for every $s,w$ such that $s
  \cdot w \ge n$ (where $\tilde O$ hides a $\polylog(m,n)$ factor). We
  also prove a lower bound, showing that every $\MA$ streaming
  algorithm for the \emph{Distinct Elements} problem that uses $s$
  bits of space and a proof of size $w$, satisfies $s \cdot w =
  \Omega(n)$.
    
  As a part of the proof of the lower bound for the \emph{Distinct
    Elements} problem, we show a new lower bound of $\Omega \left(
    \sqrt n \right )$ on the $\MA$ communication complexity of the
  \emph{Gap Hamming Distance} problem, and prove its tightness.
\end{abstract}

\section{Introduction}
The data stream computational model is an abstraction commonly used
for algorithms that process network traffic using sublinear space
\cite{AMS96,IW03,CCM09}. In the settings of this model, we have an
algorithm that gets a sequence of elements (typically, each element is
an integer) as input. This sequence of elements is called a \emph{data
  stream} and is usually denoted by $\sigma = (a_1, \ldots, a_m)$;
where $a_1$ is the first element, $a_2$ is the second element, and so
forth. The algorithm receives its input (a data stream)
element-by-element. After it sees each $a_i$, it no longer has an
access to elements with index that is smaller than $i$. The algorithm
is required to compute a function of the data stream, using as little
space as possible.

Among the most fundamental problems in the data stream model is the
problem of \emph{Distinct Elements}, i.e., the problem of computing
the number of distinct elements in a given data stream. The problem
has been studied extensively in the last two decades (see, for
example, \cite{AMS96,IW03,KNW10}). Its significance stems both from
the vast variety of applications that it spans (covering IP routing,
database operations and text compression, cf. \cite{M05,AMS96,GGK05}),
and due to the theoretical insight that it gives on the nature of
computation in the data stream model.

Alon at el. \cite{AMS96} have shown a lower bound of $\Omega(n)$
(where $n$ is the size of the alphabet from which the elements are
taken) on the streaming complexity of the computation of the
\emph{exact} number of distinct elements in a sufficiently long data
stream (i.e., where the length of the data stream is at least
proportional to $n$). The goal of reducing the space complexity of the
\emph{Distinct Elements} problem has led to a long line of research of
approximation algorithms for the problem, starting with the seminal
paper \cite{FM83} by Flajolet and Martin.  Recently, Kane at
el. \cite{KNW10} gave the first optimal approximation algorithm for
estimating the number of distinct elements in a data stream; for a
data stream with alphabet of size $n$, given $\epsilon > 0$ their
algorithm computes a $(1 \pm \epsilon)$ multiplicative approximation
using $O(\epsilon^{-2} + \log n)$ bits of space, with $2/3$ success
probability.

A natural approach for reducing the space complexity of streaming
algorithms, \emph{without} settling on an approximation, is by
considering a probabilistic proof system. Chakrabarti at
el. \cite{CCM09} have shown \emph{data stream with annotations}
algorithms for several data stream problems, using a probabilistic
proof system that is very similar to $\MA$. This line of work
continued in \cite{CMT10}, wherein a probabilistic proof system was
used in order to reduce the streaming complexity of numerous graph
problems. In a subsequent work \cite{CMT11}, Chakrabarti at
el. provided a practical instantiation of one of the most efficient
general-purpose construction of an interactive proof for arbitrary
computations, due to Goldwasser at el. \cite{GKR08}.

In this work, we study the power of Arthur-Merlin probabilistic proof
systems in the data stream model. We show a canonical $\AM$ streaming
algorithm for a wide class of data stream problems. The algorithm
offers a tradeoff between the length of the proof and the space
complexity that is needed to verify it. We show that the problem of
\emph{Distinct Elements} falls within the class of problems that our
canonical algorithm can handle. Thus, we give an $\AM$ streaming
algorithm for the \emph{Distinct Elements} problem. Given a data
stream of length $m$ over alphabet of size $n$, the algorithm uses
$\tilde O(s)$ space and a proof of size $\tilde O(w)$, for every $s,w$
such that $s \cdot w \ge n$ (where $\tilde O$ hides a $\polylog(m,n)$
factor).

In addition, we give a lower bound on the $\MA$ streaming complexity
of the \emph{Distinct Elements} problem. Our lower bound for
\emph{Distinct Elements} relies on a new lower bound that we prove on
the $\MA$ communication complexity of the \emph{Gap Hamming Distance}
problem.

\subsection{Arthur-Merlin Probabilistic Proof Systems}
An $\MA$ (Merlin-Arthur) proof is a probabilistic extension of the
notion of proof in complexity theory. Proofs of this type are commonly
described as an interaction between two players, usually referred to
as Merlin and Arthur. We think of Merlin as an omniscient prover, and
of Arthur as a computationally bounded verifier. Merlin is supposed to
send Arthur a valid proof for the correctness of a certain
statement. After seeing both the input and Merlin's proof, with high
probability Arthur can verify a valid proof for a correct statement,
and reject every possible alleged proof for a wrong statement.

Formally, the complexity class $\MA(T,W)$ is defined as follows:
\begin{definition}
  Let $\epsilon \ge 0$, and let $T,W : \N \to \N$ be monotone
  functions.  A language $L$ is in $\MA_\epsilon (T,W)$ if there
  exists a randomized algorithm $V$ (the verifier) that receives an
  input $x$ (denote its size by $|x|$) and a proof (sometimes called
  witness) $w$, such that,
  \begin{enumerate}
  \item \textbf{Completeness}: For every $x \in L$, there exists a
    string $w$ of length at most $W(|x|)$ that satisfies
    \begin{align*}
      \Pr[ V(x,w) = 1 ] > 1 - \epsilon.
    \end{align*}
  \item \textbf{Soundness}: For every $x \not\in L$, and for any
    string $w$ of length at most $W(|x|)$,
    \begin{align*}
      \Pr [ V(x,w) = 1 ] < \epsilon.
    \end{align*}
  \item For every $x,w$ the running time of $V$ on $(x,w)$ is at most
    $T(|x|)$.
  \end{enumerate}
  Under these notations, we refer to $T$ as the \emph{time complexity}
  of the verifier. The function $W$ is referred to as the
  \emph{length} of the proof, and the sum $T+W$ is called the $\MA$
  complexity of the algorithm.
\end{definition}

An $\AM$ proof is defined almost the same as an $\MA$ proof, except
that in $\AM$ proof systems we assume that both the prover and the
verifier have access to a common source of randomness (alternatively,
$\AM$ proof systems can be described as $\MA$ proof systems that start
with an extra round, wherein Arthur sends Merlin a random string).

The notion of $\AM$ and $\MA$ proof systems can be extended to many
computational models. In this work we consider both the communication
complexity analogue of $\MA$, wherein Alice and Bob receive a proof
that they use in order to save communication, and the data stream
analogues of $\MA$ and $\AM$, wherein the data stream algorithm
receives a proof and uses it in order to reduce the required resources
for solving a data stream problem.

Recently, probabilistic proof systems for streaming algorithms have
been used to provide an abstraction of the notion of \emph{delegation
  of computation} to a cloud (see \cite{CMT10,CMT11,CKL11}). In the
context of cloud computing, a common scenario is one where a user
receives or generates a massive amount of data, which he cannot afford
to store locally. The user can stream the data he receives to the
cloud, keeping only a short certificate of the data he
streamed. Later, when the user wants to calculate a function of that
data, the cloud can perform the calculations and send the result to
the user. However, the user cannot automatically trust the cloud (as
an error could occur during the computation, or the service provider
might not be honest). Thus the user would like to use the short
certificate that he saved in order to verify the answer that he gets
from the cloud.

\subsection{Communication Complexity and the \emph{Gap Hamming Distance} Problem}
Communication complexity is a central model in computational
complexity. In its basic setup, we have two computationally unbounded
players, Alice and Bob, holding (respectively) binary strings $x,y$ of
length $n$ each. The players need to compute a function of both of the
inputs, using the least amount of communication between them.

In this work we examine the well known communication complexity
problem of \emph{Gap Hamming Distance} ($\GHD$), wherein each of the
two parties gets an $n$ bit binary string, and together the parties
need to tell whether the Hamming distance of the strings is larger
than $\frac{n}{2} + \sqrt n$ or smaller than $\frac{n}{2} - \sqrt n$
(assuming that one of the possibilities occurs). In \cite{CR11} a
tight linear lower bound was proven on the communication complexity of
a randomized communication complexity protocol for $\GHD$. Following
\cite{CR11}, a couple of other proofs (\cite{Vid11,She11}) were given
for the aforementioned lower bound. Relying on \cite{She11}, in this
work we give a tight lower bound of $\Omega(\sqrt n)$ on the $\MA$
communication complexity of $\GHD$.

\subsection{Our Results}
The main contributions in this work are:
\begin{enumerate}
\item A canonical $\AM$ streaming algorithm for a wide class of data
  stream problems, including the \emph{Distinct Elements} problem.
\item A lower bound on the $\MA$ streaming complexity of the
  \emph{Distinct Elements} problem.
\item A tight lower bound on the $\MA$ communication complexity of the
  \emph{Gap Hamming Distance} problem.
\end{enumerate}

In order to state the results precisely, we first introduce the
following notations: given a data stream $\sigma=(a_{1},\ldots,
a_{m})$ (over alphabet $[n]$), the \emph{element indicator}
$\chi_i:[n]\to\B$ of the $i$'th element ($i \in [m]$) of the stream
$\sigma$, is the function that indicates whether a given element is in
position $i\in[m]$ of $\sigma$, i.e., $\chi_i(j) = 1$ if and only if
$a_i=j$. Furthermore, let $\chi:[n]\to\B^m$ be the \emph{element
  indicator} of $\sigma$, defined by
\begin{align*}
  \chi(j) = \big( \chi_1(j), \ldots, \chi_m(j) \big).
\end{align*}
In addition, given $n \in \N$ we define a \emph{clause} over $n$
variables $x_1, \ldots, x_n$ as a function $C: \B^n \to \B$ of the
form $(y_1 \vee y_2 \vee \ldots \vee y_n)$, where for every $i\in[n]$
the literal $y_i$ is either a variable ($x_j$), a negation of a
variable ($\neg x_j$), or one of the constants $\B$.

Equipped with the notations above, we formally state our results.  Let
$0 \le \epsilon < 1/2$. Let $\mathcal{P}$ be a data stream problem
such that for every $m,n \in \N$ there exists a set of $k = k(m,n)$
clauses $\{ C_t \}_{t\in[k]}$ over $m$ variables, and a function $\psi
: \B^k \to \Z$, such that for every data stream
$\sigma=(a_1,\ldots,a_m)$ with alphabet $[n]$,
\begin{align*}
  (1-\epsilon) \mathcal{P}(\sigma) \leq \sum_{j=1}^n \psi
  \big(C_1\circ\chi(j), \ldots, C_k\circ\chi(j) \big) \leq
  (1+\epsilon) \mathcal{P}(\sigma).
\end{align*}
Moreover, we assume that $\psi$ and $\{ C_t \}_{t\in[k]}$ are known to
the verifier\footnote{For example, $\psi$ and $\{ C_t \}_{t\in[k]}$
  can be $\polylog(m,n)$-space uniform; that is, the
  \emph{description} of $\psi$ and $\{ C_t \}_{t\in[k]}$ can be
  computed by a deterministic Turing machine that runs in
  $\polylog(m,n)$ space.}, and that there exists $B \le \poly(m,n)$
such that $\psi(x)<B$ for every $x \in \B^k$.  Given such
$\mathcal{P}$, for every $0 < \delta \leq 1$ and every $s,w\in\N$ such
that $s \cdot w \ge n$, we give an $\AM$ streaming algorithm, with
error probability $\delta$, for approximating $\mathcal{P}(\sigma)$
within a multiplicative factor of $1 \pm \epsilon$.  The algorithm
uses space $O \big( s k \cdot \polylog(m, n, \delta^{-1}) \big)$, a
proof of size $W = O \big( w k \cdot \polylog (m, n, \delta^{-1})
\big)$, and randomness complexity $\polylog(m, n, \delta^{-1})$.

We show that the aforementioned algorithm, when applied to the
\emph{Distinct Elements} problem with parameters $s,w$ such that $s
\cdot w \ge n$, yields an $\AM$ streaming algorithm for the
problem. The algorithm computes, with probability at least $2/3$, the
exact number of distinct elements in a data stream of length $m$ over
alphabet $[n]$, using space $\tilde O(s)$ and a proof of size $\tilde
O(w)$ (where $\tilde O$ hides a $\polylog(m,n)$ factor). For example,
by fixing $w = n$, we get an $\AM$ streaming algorithm for the
\emph{Distinct Elements} problem that uses only polylogarithmic space.

We note that an interesting special case of the class of problems that
our canonical $\AM$ streaming algorithm handles can also be stated in
terms of Boolean circuits, instead of clauses. That is, given $0 \le
\epsilon < 1/2$ and a data stream problem $\mathcal{P}$ such that for
every $m,n \in \N$ there exists an unbounded fan-in Boolean circuit
$C:\B^m \to \B$ with $k = k(m,n)$ non-input gates, such that for every
data stream $\sigma=(a_1,\ldots,a_m)$ with alphabet $[n]$,
\begin{align*}
  (1-\epsilon) \mathcal{P}(\sigma) \leq \sum_{j=1}^n C \big(\chi_1(j),
  \ldots, \chi_m(j) \big) \leq (1+\epsilon) \mathcal{P}(\sigma).
\end{align*}
Assuming that $C$ is known to the verifier, we get an $\AM$ streaming
algorithm for $\mathcal{P}$ with the same parameters as in the
original formulation of the canonical $\AM$ algorithm above.

Our next result is a lower bound on the $\MA$ streaming complexity of
the \emph{Distinct Elements} problem. We show that every $\MA$
streaming algorithm that approximates, within a multiplicative factor
of $1 \pm \nicefrac{1}{\sqrt{n}}$, the number of distinct elements in
a data stream of length $m$ over alphabet $[n]$, using $s$ bits of
space and a proof of size $w$, must satisfy $s \cdot w = \Omega(n)$.

Last, we show a tight (up to a logarithmic factor) lower bound on the
$\MA$ communication complexity of the \emph{Gap Hamming Distance}
problem. For every $\MA$ communication complexity protocol for $\GHD$
that communicates $t$ bits and uses a proof of size $w$, we have $t
\cdot w = \Omega(n)$. We prove the tightness of the lower bound by
giving, for every $t,w\in\N$ such that $t \cdot w \ge n$, an $\MA$
communication complexity protocol for $\GHD$, which communicates $O(t
\log n)$ bits and uses a proof of size $O(w \log n)$.

\subsection{Techniques}
The main intuition behind our canonical $\AM$ streaming algorithm is
based on the ``algebrization'' inspired communication complexity
protocol of Aaronson and Wigderson \cite{AW09}. However our proof is
much more technically involved.

In general, say we have a data stream problem $\mathcal{P}$ and two
integers $s,w$ such that $s \cdot w \ge n$. If there exists a low
degree polynomial $g(x,y):\Z^2\to\Z$ (that depends on the input stream
$\sigma$) and two domains $\mathcal{D}_w, \mathcal{D}_s \subseteq \Z$
of cardinality $w,s$ (respectively) such that
\begin{align*}
  \mathcal{P}(\sigma) = \sum_{x\in\mathcal{D}_w}
  \sum_{y\in\mathcal{D}_s} g(x,y),
\end{align*}
then assuming we can efficiently evaluate $g$ at a random point, by a
straightforward adaptation of the \cite{AW09} protocol to the settings
of streaming algorithms, we obtain a simple $\MA$ streaming algorithm
for $\mathcal{P}$.

However, in our case we can only express $\mathcal{P}(\sigma)$ as
\begin{align*}
  \sum_{x \in \mathcal{D}_w} \sum_{y \in \mathcal{D}_s} \psi
  \big(C_1\circ\tilde{\chi}(x,y), \ldots, C_k\circ\tilde{\chi}(x,y)
  \big),
\end{align*}
where $k$ is a natural number, $\{ C_t \}_{t\in[k]}$ are clauses over
$m$ variables, $\psi : \B^k \to \Z$ is a function over the hypercube,
$\tilde{\chi}:\mathcal{D}_w \times \mathcal{D}_s \to \B^m$ is the
bivariate equivalent of the element indicator $\chi:[n]\to\B^m$, and
$\mathcal{D}_w, \mathcal{D}_s \subseteq \Z$ are domains of cardinality
$w,s$ (respectively).

The function $\psi \big(C_1\circ\tilde{\chi}(x,y), \ldots,
C_k\circ\tilde{\chi}(x,y) \big)$ is not a low degree polynomial.  We
would have liked to overcome this difficulty by using the
approximation method of \cite{ER86,S87}. The latter allows us to have
a low degree approximation of the clauses $\{ C_t \}_{t\in[k]}$, such
that with high probability (over the construction of the approximation
polynomials) we can replace the clauses with low degree polynomials,
without changing the output. The aforementioned randomized procedure
comes at a cost of turning the $\MA$ streaming algorithm to an $\AM$
streaming algorithm.

Yet, the above does not sufficiently reduces the degree of $\psi
\big(C_1\circ\tilde{\chi}(x,y), \ldots, C_k\circ\tilde{\chi}(x,y)
\big)$. This is due to the fact that the method of \cite{ER86,S87}
results with approximation polynomials over a finite field of
cardinality that is larger than $\mathcal{P}(\sigma)$. The degree of
the approximation polynomials is close to the cardinality of the
finite field, which in our case can be a large number ($\poly(m,n)$).

Instead we aim to apply the method of \cite{ER86,S87} to approximate
\begin{align*}
  \left\{ \mathcal{P}(\sigma) \pmod{q} \right\}_{q\in Q}
\end{align*}
for a set $Q$ of $\polylog(m,n)$ primes, each of size at most
$\polylog(m,n)$. This way, each approximation polynomial that we get
is over a finite field of cardinality $\polylog(m,n)$, and of
sufficiently low degree. Then, we use the \emph{Chinese Remainder
  Theorem} to extract the value of $\mathcal{P}(\sigma)$ from $\left\{
  \mathcal{P}(\sigma) \pmod{q} \right\}_{q\in Q}$.

Nonetheless, this is still not enough, as for every $q\in Q$ we want
the answer to be the summation of the polynomial approximation of
$\psi \big(C_1\circ\tilde{\chi}(x,y), \ldots,
C_k\circ\tilde{\chi}(x,y) \big) \pmod{q}$ over some domain
$\mathcal{D}_w \times \mathcal{D}_s \subseteq \Z^2$ (where
$|\mathcal{D}_w| = w$ and $|\mathcal{D}_s| = s$). Since the
cardinality of the field $\F_q$ is typically smaller than $w$ and $s$,
we use an extension (of sufficient cardinality) of the field $\F_q$.

At each step of the construction, we make sure that we perserve both
the restrictions that are imposed by the data stream model, and the
conditions that are needed to ensure an efficient verification of the
proof.

The idea behind our $\AM$ streaming algorithm for \emph{Distinct
  Elements} is simply noting that we can indicate whether an element
$j$ appears in the data stream, by the disjunction of the element
indicators of $j\in[n]$ in all of the positions of the stream (i.e.,
$\chi_1(j), \ldots, \chi_m(j)$). Then we can represent the number of
distinct elements as a sum of disjunctions, and use the canonical
$\AM$ streaming algorithm in order to solve the \emph{Distinct
  Elements} problem.

As for the lower bound on the $\MA$ streaming complexity of the
\emph{Distinct Elements} problem, we start by establishing a lower
bound on the $\MA$ communication complexity of the \emph{Gap Hamming
  Distance} problem ($\GHD$). A key element in the proof of the latter
is based on Sherstov's recent result \cite{She11} on the \emph{Gap
  Orthogonality} problem ($\ORT$) and its relation to $\GHD$. Sherstov
observed that the problem of \emph{Gap Orthogonality} readily reduces
to \emph{Gap Hamming Distance} problem. Although at first glance it
seems that the transition to $\ORT$ is of little substance, it turns
out that Yao's corruption bound \cite{Y83} suits it perfectly. In
fact, the corruption property for $\ORT$ is equivalent to the
anti-concentration property of orthogonal vectors in the Boolean
cube. Using this observation, we prove a lower bound on the $\MA$
communication complexity of $\ORT$ (following the method of
\cite{RS04}), which in turn, by the reduction from $\ORT$ to $\GHD$,
implies a lower bound on the $\MA$ communication complexity of
$\GHD$. Next we adapt the reduction that was implicitly stated in
\cite{IW03}, and reduce the $\MA$ communication complexity problem of
$\GHD$ to the $\MA$ problem of calculating the exact number of
\emph{Distinct Elements}.

\subsection{Related Work}
The data stream model has gained a great deal of attention after the
publication of the seminal paper by Alon, Matias and Szegedy
\cite{AMS96}. In the scope of that work, the authors have shown a
lower bound of $\Omega(n)$ (where $n$ is the size of the alphabet) on
the streaming complexity of \emph{Distinct Elements} (i.e., the
computation of the \emph{exact} number of distinct elements in a data
stream) where the length of the input is at least proportional to $n$.

Following \cite{AMS96} there was a long line of theoretical research
on the approximation of the \emph{Distinct Element} problem
(\cite{YJK02,IW03,BGJ07,BC09,KNW10}, see \cite{M05} for a survey of
earlier results). Finally, Kane at el. \cite{KNW10} gave the first
optimal approximation algorithm for estimating the number of distinct
elements in a data stream; for a data stream with alphabet of size
$n$, given $\epsilon>0$ their algorithm computes a $(1 \pm \epsilon)$
multiplicative approximation using $O(\epsilon^{-2} + \log n)$ bits of
space, with $2/3$ success probability. This result matches the tight
lower bound of Indyk and Woodruff \cite{IW03}.

In a recent sequence of works, the data stream model was extended to
support several interactive and non-interactive proof systems
\cite{CCM09,CMT10,CKL11}.  The model of streaming algorithms with
non-interactive proofs was first introduced in \cite{CCM09} and
extended in \cite{CMT10,CMT11}. In \cite{CCM09} the authors gave an
optimal (up to polylogarithmic factors) \emph{data stream with
  annotations} algorithm for computing the $k$'th frequency moment
exactly, for every integer $k \ge 1$.

\section{Preliminaries}
\subsection{Communication Complexity}
Let $X,Y,Z$ be finite sets, and let $f : X \times Y \to Z$ be a
(possibly partial) function. In the \emph{two-party probabilistic
  communication complexity model} we have two computationally
unbounded players, traditionally referred to as Alice and Bob. Both
players share a random string. Alice gets as an input $x \in X$. Bob
gets as an input $y \in Y$. At the beginning, none of the players has
any information regarding the input of the other player. Their common
goal is to compute the value of $f(x,y)$, using a protocol that
communicates as small number of bits as possible. In each step of the
protocol, one of the players sends one bit to the other player. This
bit may depend on the player's input, the common random string, as
well as on all previous bits communicated between the two players. At
the end of the protocol, both players have to know the value of
$f(x,y)$ with high probability.

\subsubsection{MA Communication Complexity}
In $\MA$ \emph{communication complexity protocols}, we have a
(possibly partial) function $f : X \times Y \to \B$ (for some finite
sets $X,Y$), and three computationally unbounded parties: Merlin,
Alice, and Bob. The function $f$ is known to all parties. Alice gets
as an input $x \in X$. Bob gets as an input $y \in Y$. Merlin sees
both $x$ and $y$. We think of Merlin as a \emph{prover}, and think of
Alice and Bob as \emph{verifiers}. We assume that Alice and Bob share
a private random string that Merlin cannot see.

At the beginning of an $\MA$ \emph{communication complexity protocol},
Merlin sends a proof string $w$ to both Alice and Bob, so both players
have a free access to $w$. The players proceed as before.  In each
step of the protocol, one of the players sends one bit to the other
player. At the end of the protocol, both players have to know an
answer $z$. Hence, the answer depends on the input $(x,y)$ as well as
on the proof $w$. For a protocol $P$, denote by $P \big( (x,y), w
\big)$ the probabilistic answer $z$ given by the protocol on input
$(x,y)$ and proof $w$.

An $\MA$ \emph{communication complexity protocol} has three
parameters: a limit on the probability of error of the protocol,
denoted by $\epsilon$; a limit on the number of bits of communication
between Alice and Bob, denoted by $T$; and a limit on the length of
Merlin's proof string, denoted by $W$.

With the above in mind, we can now define $\MA_\epsilon (T,W)$
communication complexity as follows:

\begin{definition}
  An $\MA_\epsilon (T,W)$-communication complexity protocol for $f$ is
  a probabilistic communication complexity protocol $P$, as above
  (i.e., with an additional proof string $w$ presented to the
  players). During the protocol, Alice and Bob communicate at most $T$
  bits. The protocol satisfies,
  \begin{enumerate}
  \item \textbf{Completeness}: for all $(x,y)\in f^{-1}(1)$, there
    exists a string $w$ such that $|w|<W$, that satisfies
    \begin{align*}
      \Pr \left[ P \big( (x,y), w \big) = 1 \right] > 1 - \epsilon.
    \end{align*}
  \item \textbf{Soundness}: for all $(x,y)\in f^{-1}(0)$ and for any
    string $w$ such that $|w|<W$, we have
    \begin{align*}
      \Pr \left[ P \big( (x,y) , w \big) = 1 \right] < \epsilon.
    \end{align*}
  \end{enumerate}
\end{definition}

\subsubsection{The \emph{Gap Hamming Distance} Problem}
Let $n \in \N$, and let $\zeta_0,\zeta_1>0$. We define the \emph{Gap
  Hamming Distance} problem as follows:

\begin{definition}
  The \emph{Gap Hamming Distance} problem is the communication
  complexity problem of computing the partial Boolean function
  $\GHD_{n,\zeta_0,\zeta_1}:\BM^n\times\BM^n\to \B$ given by
  \begin{align*}
    \GHD_{n,\zeta_0,\zeta_1} (x,y)=
    \begin{cases}
      1 & if\quad \langle x,y\rangle > \zeta_1 \\
      0 & if\quad \langle x,y\rangle < - \zeta_0 \\
    \end{cases}
  \end{align*}

  Denote $\GHD=\GHD_{n,\sqrt{n},\sqrt{n}}$.
\end{definition}

\subsection{Streaming Complexity}
Let $\epsilon \ge 0$, $\delta>0$. Let $m,n\in\N$. A \emph{data stream}
$\sigma=(a_{1},\ldots, a_{m})$ is a sequence of elements, each from
$[n] = \{1,\ldots,n\}$. We say that the \emph{length} of the stream is
$m$, and the \emph{alphabet size} is $n$.

A \emph{streaming algorithm} is a space-bounded probabilistic
algorithm that gets an element-by-element access to a \emph{data
  stream}. After each element arrives, the algorithm can no longer
access the elements that precede it. At the end of its run, the
\emph{streaming algorithm} is required to output (with high
probability) a certain function of the \emph{data stream} that it
read. When dealing with \emph{streaming algorithms}, the main resource
we are concerned with is the size of the space that the algorithm
uses.

Formally, a \emph{data stream problem} $\mathcal{P}$ is a collection
of functions $\{ f_{m,n}:[n]^{m}\to\R \}_{m,n \in \N}$. That is, a
function for every combination of \emph{length} and \emph{alphabet
  size} of a \emph{data stream}. However, slightly abusing notation
for the sake of brevity, we will define each \emph{data stream
  problem} by a single function (which in fact depends on the length
$m$ and alphabet size $n$ of the \emph{data stream}). A
$\delta$-error, $\epsilon$-approximation data stream algorithm
$\mathcal{A}_{\epsilon,\delta}$ for $\mathcal{P}$ is a probabilistic
algorithm that gets a sequential, one pass access to a \emph{data
  stream} $\sigma=(a_{1},\ldots, a_{m})$ (where each $a_{i}$ is a
member of $[n]$), and satisfies:
\begin{align*}
  \Pr\left[ \left|
      \frac{\mathcal{A_{\epsilon,\delta}}(\sigma)}{f_{m,n}(\sigma)} -
      1 \right| > \epsilon \right] < \delta.
\end{align*}
If $\epsilon=0$ we say that the \emph{streaming algorithm} is exact.

Last, given a data stream problem $\mathcal{P} = \{ f_{m,n}:[n]^m \to
\R \}_{m,n \in \N}$ and a data stream $\sigma = (a_1, \ldots, a_m)$
(with alphabet $[n]$) we denote by $\mathcal{P}(\sigma)$ the output of
$f_{m,n}(\sigma)$, for the $f_{m,n} \in \mathcal{P}$ that matches the
length and alphabet size of $\sigma$. Similarly, when applying a
family of functions to $\sigma$, we in fact apply a specific function
in the family, according to the parameters $m,n$ of $\sigma$.

\subsubsection{The \emph{Distinct Elements} Problem}
The \emph{Distinct Elements} problem is the problem of computing the
exact number of distinct elements that appear in a data stream,
denoted by $F_0(\sigma)$. Formally, we define:
\begin{definition}
  The \emph{Distinct Elements} problem is the data stream problem of
  computing the exact number of distinct elements in a data stream
  $\sigma=(a_{1},\ldots, a_{m})$ (where $a_{i} \in [n]$ for every
  $i$), i.e., computing (exactly):
  \begin{align*}
    F_0(\sigma) = \big |\left\{ i\in\N \; : \; \exists j\in[m] \;\;
      a_{j}=i \right\} \big|.
  \end{align*}
\end{definition}
Note that if we define $0^0=0$ then this is exactly the $0$'th
frequency moment of the stream. Hence the notation $F_0$.

\section{Streaming Algorithms with Probabilistic Proof Systems}
\label{sec:pps}
In this section we extend the data stream computational model in order
to support two types of probabilistic proof systems: $\MA$ algorithms,
wherein the streaming algorithm gets a proof that it probabilistically
verifies, and $\AM$ algorithms that extend $\MA$ algorithms by adding
shared randomness. We study both of these probabilistic proof systems
in two variations: in the first, the proof is also being streamed to
the verifier, and in the second, the verifier has a free access to the
proof. Formal definitions follow.

\subsection{MA Streaming Algorithms}
Similarly to the way $\MA$ communication complexity protocols are
defined, in $\MA$ streaming algorithms we have an omniscient prover
(Merlin) who sends a proof to a verifier (Arthur), which is in fact a
streaming algorithm that gets both the input stream and the proof
(either by a free access or by a one-pass, sequential access). The
streaming algorithm computes a function of the input stream. Using the
proof we hope to achieve a better space complexity than what the
regular streaming model allows.

We start with $\MA$ proofs wherein the proof is being streamed to the
verifier. Formally, we define

\begin{definition}
  Let $\epsilon \ge 0$, $\delta>0$, and let $\mathcal{P} = \{
  f_{m,n}:[n]^m\to \R \}_{m,n \in \N}$ be a data stream problem. An
  $\MA$ streaming algorithm for $\mathcal{P}$ is a probabilistic data
  stream algorithm $\mathcal{A}$, which simultaneously gets two
  streams: an input stream $\sigma=(a_{1},\ldots, a_{m})$ (where
  $a_{i} \in [n]$ for every $i$) and a proof stream $\omega$; to both
  it has a sequential, one pass access. Given two functions $S,W :
  \N^2 \to \N$, we say that an $\MA$ streaming algorithm is
  $\MA_{\epsilon,\delta} \big( S(m,n), W(m,n) \big)$ if it uses at
  most $S(m,n)$ bits of space, and satisfies:
  \begin{enumerate}
  \item \textbf{Completeness}: for every $\sigma=(a_{1},\ldots,
    a_{m})$ (with \emph{alphabet} $[n]$) there exists a non empty set
    $\mathcal{W_\sigma}$ of proof streams of length at most $W(m,n)$,
    such that for every $\omega\in\mathcal{W_\sigma}$ we have,
    \begin{align*}
      \Pr \left[ \left|
          \frac{\mathcal{A}(\sigma,\omega)}{f_{m,n}(\sigma)} -1
        \right| \le \epsilon \right] >1-\delta
    \end{align*}
  \item \textbf{Soundness}: for every $\sigma=(a_{1},\ldots, a_{m})$
    (with \emph{alphabet} $[n]$), and for every
    $\omega\not\in\mathcal{W_\sigma}$ we have
    \begin{align*}
      \Pr[\mathcal{A}(\sigma,\omega)\neq \bot]<\delta
    \end{align*}
    where $\bot\not\in\R$ is a symbol that represents that the
    algorithm could not verify the correctness of the proof.
  \end{enumerate}
\end{definition}

The second natural way to define an $\MA$ probabilistic proof system
for the data stream model, is by allowing the algorithm a free access
to the proof. This leads to the following definition:

\begin{definition}
  Let $\epsilon \ge 0$, $\delta>0$, and let $\mathcal{P} = \{
  f_{m,n}:[n]^m\to \R \}_{m,n \in \N}$ be a data stream problem. An
  $\widehat{\MA}$ streaming algorithm for $\mathcal{P}$ is a
  probabilistic data stream algorithm $\mathcal{A}^w$, which has a
  free oracle access to a proof string $w$. The algorithm gets a
  stream $\sigma=(a_{1},\ldots, a_{m})$ (where $a_{i} \in [n]$ for
  every $i$) as an input, to which it has a sequential, one pass
  access. Given two functions $S,W : \N^2 \to \N$, we say that an
  $\widehat{\MA}$ streaming algorithm is
  $\widehat{\MA}_{\epsilon,\delta} \big( S(m,n), W(m,n) \big)$ if it
  uses at most $S(m,n)$ bits of space, and satisfies:
  \begin{enumerate}
  \item \textbf{Completeness}: for every $\sigma=(a_{1},\ldots,
    a_{m})$ (with \emph{alphabet} $[n]$), there exists a non empty set
    $\mathcal{W}_\sigma$ of proof strings of length at most $W(m,n)$,
    such that for every $w\in\mathcal{W}_\sigma$ we have,
    \begin{align*}
      \Pr \left[ \left| \frac{\mathcal{A}^w(\sigma)}{f_{m,n}(\sigma)}
          -1 \right| \le \epsilon \right] >1-\delta
    \end{align*}
  \item \textbf{Soundness}: for every $\sigma=(a_{1},\ldots, a_{m})$
    (with \emph{alphabet} $[n]$), and for every
    $w\not\in\mathcal{W}_\sigma$ we have
    \begin{align*}
      \Pr[\mathcal{A}^w(\sigma)\neq \bot]<\delta
    \end{align*}
    where $\bot\not\in\R$ is a symbol that represents that the
    algorithm could not verify the correctness of the proof.
  \end{enumerate}
\end{definition}

Note that by definition, the model of $\MA$ streaming with a free
access to the proof is stronger than the model of $\MA$ streaming with
a proof stream. Thus when in Section~\ref{sec:distinctelements} we
prove lower bounds on the $\widehat{\MA}$ streaming complexity, it
also implies lower bounds on the $\MA$ streaming complexity.

\subsection{AM Streaming Algorithms}
We can further extend the data stream model to support an $\AM$
probabilistic proof system. Similarly to the case of $\MA$ proofs, an
$\AM$ streaming algorithm receives a proof stream and an input stream,
to which it has a sequential, one pass access; except that in $\AM$
proof systems the prover and verifier also share a common random
string. Formally, we define
\begin{definition}
  Let $\epsilon \ge 0$, $\delta>0$, and let $\mathcal{P} = \{
  f_{m,n}:[n]^m\to \R \}_{m,n \in \N}$ be a data stream problem. An
  $\AM$ streaming algorithm for $\mathcal{P}$ is a probabilistic data
  stream algorithm $\mathcal{A}^r$ that has an oracle access to a
  common random string $r$, and that is also allowed to make private
  random coin tosses. The algorithm simultaneously gets two streams:
  an input stream $\sigma=(a_{1},\ldots, a_{m})$ (where $a_{i} \in
  [n]$ for every $i$) and a proof stream $\omega$, to both it has a
  sequential, one pass access. Given two functions $S,W : \N^2 \to
  \N$, we say that an $\AM$ streaming algorithm is
  $\AM_{\epsilon,\delta} \big( S(m,n), W(m,n) \big)$ if it uses at
  most $S(m,n)$ bits of space, and satisfies that for every
  $\sigma=(a_{1},\ldots, a_{m})$ (over \emph{alphabet} $[n]$), with
  probability at least $1 - \delta/2$ (over $r$) there exists a non
  empty set $\mathcal{W}_\sigma(r)$ of proof streams of length at most
  $W(m,n)$, such that:
  \begin{enumerate}
  \item \textbf{Completeness}: For every
    $\omega\in\mathcal{W}_\sigma(r)$
    \begin{align*}
      \Pr \left[ \left|
          \frac{\mathcal{A}^r(\sigma,\omega)}{f_{m,n}(\sigma)} -1
        \right| \le \epsilon \right] > 1 - \frac{\delta}{2},
    \end{align*}
    where the probability is taken over the private random coin tosses
    of $\mathcal{A}^r$.
  \item \textbf{Soundness}: For $\omega\not\in\mathcal{W}_\sigma(r)$
    \begin{align*}
      \Pr \left[ \mathcal{A}^r(\sigma,\omega) = \bot \right] > 1 -
      \frac{\delta}{2},
    \end{align*}
    where the probability is taken over the private random coin tosses
    of $\mathcal{A}^r$, and $\bot\not\in\R$ is a symbol that
    represents that the algorithm could not verify the correctness of
    the proof.
  \end{enumerate}
  The \emph{randomness complexity} of the algorithm is the total size
  of the common random string $r$, and the number of private random
  coin tosses that the algorithms performs.
\end{definition}
Note that we slightly deviate from the standard definition of an $\AM$
algorithm, by allowing $\mathcal{A}$ to be a probabilistic algorithm
with a private random string.

Just as with the $\MA$ streaming model, we can define $\widehat{\AM}$
streaming algorithms by allowing a free access to the proof. Again, by
definition the model of $\AM$ streaming with a free access to the
proof is stronger than the model of $\AM$ streaming with a proof
stream. Our canonical $\AM$ algorithm works for the weaker model,
wherein the proof is being streamed, thus our $\AM$ upper bounds also
implies $\widehat{\AM}$ upper bounds.

\textbf{Note 1:} In both of the models ($\MA$ and $\AM$), as
traditionally done in Arthur-Merlin probabilistic proof systems, we
will sometimes describe the $\MA$/$\AM$ algorithm as an interaction
between an omniscient prover Merlin, who sends an alleged proof of a
statement to Arthur, a computationally limited verifier (in our case,
a streaming algorithm), who in turn probabilistically verifies the
correctness of Merlin's proof.

\textbf{Note 2:} In all of our ($\MA$ and $\AM$) algorithms, we assume
without loss of generality that Arthur knows both the length $m$ and
the alphabet size $n$. This can be done since we can insert $m,n$ at
the beginning of the proof. Then, Arthur only needs to verify that the
length of the stream was indeed $m$, and that no element was bigger
than $n$. Since all of the algorithms we present in this paper are
$\Omega(\log m + \log n)$ in both proof size and space complexity,
this does not change their overall asymptotical complexity.

\section{The Canonical AM Streaming Algorithm}
\label{sec:canonical}
In this section we show our canonical $\AM$ algorithm. Recall that
given a data stream $\sigma=(a_{1},\ldots, a_{m})$ (over alphabet
$[n]$), the \emph{element indicator} $\chi_i:[n]\to\B$ of the $i$'th
element ($i \in [m]$) of the stream $\sigma$, is the function that
indicates whether a given element is in position $i\in[m]$ of
$\sigma$, i.e., $\chi_i(j) = 1$ if and only if $a_i=j$. Furthermore,
let $\chi:[n]\to\B^m$ be the \emph{element indicator} of $\sigma$,
defined by
\begin{align*}
  \chi(j) = \big( \chi_1(j), \ldots, \chi_m(j) \big).
\end{align*}
In addition, given $n \in \N$ we define a \emph{clause} over $n$
variables $x_1, \ldots, x_n$ as a function $C: \B^n \to \B$ of the
form $(y_1 \vee y_2 \vee \ldots \vee y_n)$, where for every $i\in[n]$
the literal $y_i$ is either a variable ($x_j$), a negation of a
variable ($\neg x_j$), or one of the constants $\B$.

We prove the following theorem:
\begin{theorem}
  \label{thm:main}
  Let $0 \le \epsilon < 1/2$. Let $\mathcal{P}$ be a data stream
  problem such that for every $m,n \in \N$ there exists a set of $k =
  k(m,n)$ clauses $\{ C_t \}_{t\in[k]}$ over $m$ variables, and a
  function $\psi : \B^k \to \Z$, such that for every data stream
  $\sigma=(a_1,\ldots,a_m)$ with alphabet $[n]$,
  \begin{align*}
    (1-\epsilon) \mathcal{P}(\sigma) \leq \sum_{j=1}^n \psi
    \big(C_1\circ\chi(j), \ldots, C_k\circ\chi(j) \big) \leq
    (1+\epsilon) \mathcal{P}(\sigma).
  \end{align*}
  Moreover, we assume that $\psi$ and $\{ C_t \}_{t\in[k]}$ are known
  to the verifier, and that there exists $B \le \poly(m,n)$ such that
  $\psi(x)<B$ for every $x \in \B^k$. Then, for every $0 < \delta \leq
  1$ and every $s,w\in\N$ such that $s \cdot w \ge n$, there exists an
  explicit $\AM_{\epsilon,\delta}(S,W)$-streaming algorithm for
  approximating $\mathcal{P}(\sigma)$; where $S = O \big( s k \cdot
  \polylog(m, n, \delta^{-1}) \big)$, $W = O \big( w k \cdot \polylog
  (m, n, \delta^{-1}) \big)$, and the randomness complexity is
  $\polylog(m, n, \delta^{-1})$.
\end{theorem}

\begin{proof}
  Let $0 \le \epsilon < 1/2$. Let $\mathcal{P}$ be a data stream
  problem such that for every $m,n \in \N$ there exists a set of $k =
  k(m,n)$ clauses $\{ C_t \}_{t\in[k]}$ over $m$ variables, and a
  function $\psi : \B^k \to \Z$, such that for every data stream
  $\sigma=(a_1,\ldots,a_m)$ with alphabet $[n]$,
  \begin{align}
    \label{eq:assump}
    (1-\epsilon) \mathcal{P}(\sigma) \leq \sum_{j=1}^n \psi
    \big(C_1\circ\chi(j), \ldots, C_k\circ\chi(j) \big) \leq
    (1+\epsilon) \mathcal{P}(\sigma).
  \end{align}
  Assume that $\psi$ and $\{ C_t \}_{t\in[k]}$ are known to the
  verifier, and that there exists $B \le \poly(m,n)$ such that
  $\psi(x)<B$ for every $x \in \B^k$. Observe that since $\psi$ gets
  $\B$ values as inputs, we can think of $\psi$ as a multilinear
  polynomial. Assume without loss of generality that $k \le m$
  (otherwise the theorem follows trivially).

  Let $0 < \delta \leq 1$ and let $s,w\in\N$ such that $s \cdot w \ge
  n$ (assume for simplicity and without loss of generality that $s
  \cdot w = n$ exactly). We show that there exists an explicit
  $\AM_{\epsilon,\delta}(S,W)$-streaming algorithm for approximating
  $\mathcal{P}(\sigma)$; where
  \begin{align*}
    S & = O \big( s k \cdot \polylog(m, n, \delta^{-1}) \big), \\
    W & = O \big( w k \cdot \polylog (m, n, \delta^{-1}) \big),
  \end{align*}
  and the randomness complexity is $\polylog (m, n, \delta^{-1})$.

  Let $\sigma=(a_1,\ldots,a_m)$ be a data stream with alphabet
  $[n]$. The first step is representing the middle term of
  (\ref{eq:assump}) as a summation of a low degree polynomial over
  some domain. Specifically, we represent the element indicators $\{
  \chi_i \}_{i\in[m]}$ as bivariate polynomials over a finite field.

  Let $p$ be a sufficiently large (to be determined later) prime
  number of order
  \begin{align*}
    \frac{1}{\delta} \cdot \poly(m,n)
  \end{align*}
  such that: $p > 2nB > \mathcal{P}(\sigma)$. Let
  $\mathcal{D}_s(\F_p)$ be any efficiently enumerable subset, of
  cardinality $s$, of the field $\F_p$ (e.g., the lexicographically
  first elements in some representation of the field
  $\F_p$). Likewise, let $\mathcal{D}_w(\F_p)$ be any efficiently
  enumerable subset, of cardinality $w$, of the field $\F_p$. Note
  that since $n=w \cdot s$, there exists a one-to-one mapping between
  the domain $[n]$ and the domain $\mathcal{D}_w(\F_p)\times
  \mathcal{D}_s(\F_p)$. Fix such (efficiently computable) mapping
  $\pi:[n] \to \mathcal{D}_w(\F_p)\times \mathcal{D}_s(\F_p)$ (e.g.,
  according to the lexicographic order).

  For every $i \in [m]$ we can view $\chi_i:[n]\to\B$ as a bivariate
  polynomial $\tilde{\chi}_i:\F_p^2 \to \F_p$ of degree $w-1$ in the
  first variable (which we denote by $x$), and degree $s-1$ in the
  second variable (which we denote by $y$), such that for every
  $j\in[n]$ we have $\tilde{\chi}_i \circ \pi(j) = \chi_i(j)$. If we
  denote $(\alpha_i,\beta_i) \coloneqq \pi(a_i)$, then the extension
  $\tilde{\chi}_i:\F_p^2 \to \F_p$ is given explicitly by the Lagrange
  interpolation polynomial:
  \begin{align}
    \label{eq:lagrange}
    \tilde{\chi}_i(x,y) = \frac{\displaystyle \prod_{\substack{a \in
          \mathcal{D}_w(\F_p) \\ a \neq \alpha_i}} (x-a)
      \prod_{\substack{ b \in \mathcal{D}_s(\F_p) \\ b \neq \beta_i}}
      (y-b)} {\displaystyle \prod_{\substack{a \in \mathcal{D}_w(\F_p)
          \\ a \neq \alpha_i}} (\alpha_i-a) \prod_{\substack{ b \in
          \mathcal{D}_s(\F_p) \\ b \neq \beta_i}} (\beta_i-b)}
  \end{align}
  Note that for every $\xi \in \mathcal{D}_s(\F_p)$, the degree of the
  univariate polynomial $\tilde{\chi}_i(\cdot, \xi): \F_p \to \F_p$ is
  at most $w-1$.

  Let $\tilde{\chi} : \F_p^2 \to \F_p^m$ be the polynomial extension
  of the \emph{element indicator} of $\sigma$, defined by
  \begin{align*}
    \tilde{\chi}(x,y) = \big( \tilde{\chi}_1(x,y), \ldots,
    \tilde{\chi}_m(x,y) \big).
  \end{align*}
  Plugging-in the polynomial extensions of the element indicators to
  (\ref{eq:assump}) yields that
  \begin{align}
    \label{eq:bivariate_representation}
    \widetilde{\mathcal{P}}(\sigma) \coloneqq \sum_{x \in
      \mathcal{D}_w(\F_p) } \sum_{y \in \mathcal{D}_s(\F_p) } \psi
    \big(C_1\circ\tilde{\chi}(x,y), \ldots, C_k\circ\tilde{\chi}(x,y)
    \big)
  \end{align}
  (where the summation is over $\Z$) approximates
  $\mathcal{P}(\sigma)$ within a multiplicative factor of $1 \pm
  \epsilon$. Later, we will give analogous expressions of
  $\mathcal{P}(\sigma) \pmod{q}$ for prime numbers $q = O(\log p)$.

  Next, we replace each clause in (\ref{eq:bivariate_representation})
  with a low degree polynomial (over a small finite field) that
  approximates it. Towards this end, we show the following lemma
  (originated in \cite{ER86,S87}):
  \begin{lemma}
    \label{lem:circuit_approx}
    Let $\delta'>0$, let $q$ be a prime number, and let $\{ C_t
    \}_{t\in[k]}$ be a set of $k$ clauses over $m$ variables. Using
    $\polylog(m, k, \delta'^{-1})$ random coin flips, we can construct
    a set of polynomials $\{p_t:\F_q^m \to \F_q\}_{t\in[k]}$ of degree
    $O(q \log \nicefrac{k} {\delta'})$ each, such that for every $x
    \in \B^m$,
    \begin{align*}
      \Pr\left[ \forall t\in[k] \quad p_t(x) = C_t(x) \right] \geq 1 -
      \delta'
    \end{align*}
    (where the probability is taken over the random coin flips
    performed during the construction of $\{p_t\}_{t\in[k]}$).
  \end{lemma}

  \begin{proof}
    Consider $\mathcal{C} \coloneqq \{ C_t \}_{t\in[k]}$, where for
    every $t\in[k]$, $C_t$ is a clause over $m$ variables. We
    approximate each $C_t\in\mathcal{C}$ by a polynomial $p_t:\F_q^m
    \to \F_q$. Recall that every clause in $\mathcal{C}$ is an
    $m$-variate disjunction gate that operates on literals, which are
    either a variable, or a negation of a variable, or one of the
    constants $\B$.

    In order to construct a polynomial approximation of a clause
    $C_t\in\mathcal{C}$, we first replace each negation gate over a
    variable $x$ in $C_t$, with the polynomial $1-x$. Note that this
    polynomial computes the negation exactly (i.e., no approximation).

    Next, we use the method of \cite{ER86,S87} to approximate the
    $m$-variate disjunction gate of $C_t$, by constructing an
    approximation polynomial in the following way: let $\ECC:\F_q^m\to
    \F_q^{100m}$ be a linear error correcting code with relative
    distance $1/3$. Fix
    \begin{align*}
      L = O \left( \log k + \log\frac{1}{\delta'} \right),
    \end{align*}
    such that
    \begin{align*}
      \left( \frac{2}{3} \right)^L \leq \frac{\delta'}{k},
    \end{align*}
    and choose independently and uniformly at random
    $\iota_1,\ldots,\iota_L\in[100m]$. We build a low degree
    polynomial approximation for the Boolean disjunction
    function. Consider $\eta : \F_q^m \to \F_q$, defined by
    \begin{align*}
      \eta(z_1,\ldots,z_m) = 1 - \prod_{l=1}^{L}\Big(1 - \big(\ECC
      (z_1,\ldots,z_m)_{\iota_l}\big)^{q-1} \Big).
    \end{align*}
    Since $\ECC$ is linear, $\eta$ is a polynomial of degree $O(L
    \cdot q)$ in the variables $z_1, \ldots, z_m$. Observe that the
    linearity and the relative distance of $\ECC$, together with
    Fermat's little theorem implies that for every $(x_1, \ldots, x_m)
    \in \B^m$,
    \begin{align}
      \label{eq:ecc_approx}
      \Pr\left[ \eta(x_1,\ldots,x_m) \neq \bigvee_{i=1}^m x_i \right]
      \leq \left( \frac{2}{3} \right)^L \leq \frac{\delta'}{k}
    \end{align}
    (where the probability is taken over the random choices of
    $\iota_1,\ldots,\iota_L \in [100m]$). Note that we use the same
    polynomial $\eta$ for all of the clauses in $\mathcal{C}$. Thus,
    the total number of coin flips that we use is $\polylog(m, k,
    \delta'^{-1})$. The last step of the construction is defining
    $p_t$ as the composition of the disjunction polynomial $\eta$ and
    the literals in the clause $C_t$.

    Note that applying the approximation procedure that we described
    above to all of the clauses in $\mathcal{C}$, results with a set
    of $k$ polynomials $\{p_t: \F_q^m\to \F_q\}_{t\in[k]}$, where for
    every $t\in[k]$ the degree of $p_t$ is $O(q \log
    \nicefrac{k}{\delta'} )$. We conclude the proof of the lemma by
    noticing that (\ref{eq:ecc_approx}) together with a union bound
    imply that for every $x \in \B^m$,
    \begin{align*}
      \Pr\left[\forall t\in[k] \quad p_t(x) = C_t(x) \right] \geq 1 -
      \delta'
    \end{align*}
    (where the probability is over the random choices of
    $\iota_1,\ldots,\iota_L \in [100m]$).
  \end{proof}

  Observe that by applying Lemma~\ref{lem:circuit_approx} with
  $\delta' = \delta$ and $p$ as the prime number, we can represent
  (\ref{eq:bivariate_representation}) as a summation over a
  polynomial. However, the degree of this polynomial (which is
  dominated by $p$), is too high for our needs. Instead, we
  approximate (\ref{eq:bivariate_representation}) by $O(\log p)$ low
  degree polynomials.

  We start by introducing the necessary notations. Let $Q = \{ q_1,
  \ldots, q_{\rho(c\log p)}\}$ (where $\rho:\N\to\N$ is the prime
  counting function) be the set of all prime numbers that are smaller
  or equal to $c\log p$, where $c$ is a constant such that
  \begin{align*}
    \prod_{q \in Q} q > p.
  \end{align*}

  For every $q\in Q$ denote $\H_q \coloneqq \F_{q^{\lambda_q}}$, where
  $\lambda_q$ is the minimum integer that satisfies $q^{\lambda_q} >
  p$.  Since $q = O(\log p)$, and by the minimality of $\lambda_q$, we
  have $|\H_q| < pq = O(p \log p)$. Furthermore,
  \begin{align}
    \label{eq:fq_representation}
    \widetilde{\mathcal{P}}(\sigma) \pmod{q} = \sum_{x \in
      \mathcal{D}_w(\F_p) } \sum_{y \in \mathcal{D}_s(\F_p) } \psi
    \big(C_1\circ\tilde{\chi}(x,y), \ldots, C_k\circ\tilde{\chi}(x,y)
    \big) \pmod{q}
  \end{align}
  (where we can think of the summation over $\Z$ modulo $q$, as
  summation over $\F_q$). Denote
  \begin{align*}
    \widetilde{\mathcal{P}}_q(\sigma) \coloneqq
    \widetilde{\mathcal{P}}(\sigma) \pmod{q}.
  \end{align*}

  Analogously to the definitions for $\F_p$; for every prime $q\in Q$
  we define efficiently enumerable subsets $\mathcal{D}_s(\H_q)$,
  $\mathcal{D}_w(\H_q)$ of $\H_q$, with cardinality $s,w$
  (respectively), and a one-to-one mapping $\pi_q:[n] \to
  \mathcal{D}_w(\H_q) \times \mathcal{D}_s(\H_q)$. For every $i \in
  [m]$, we can view $\chi_i:[n]\to\B$ as a bivariate polynomial
  $\tilde{\chi}_i^q:\H_q^2 \to \H_q$ of degree $w-1$ in the first
  variable (which we denote by $x$), and degree $s-1$ in the second
  variable (which we denote by $y$), such that for every $j\in[n]$ we
  have $\tilde{\chi}_i^q \circ \pi_q(j) = \chi_i(j)$. Let
  $\tilde{\chi}^q : \H_q^2 \to \H_q^m$ be defined by
  \begin{align*}
    \tilde{\chi}^q(x,y) = \big( \tilde{\chi}^q_1(x,y), \ldots,
    \tilde{\chi}^q_m(x,y) \big).
  \end{align*}
  Moreover, we can think of the multilinear polynomial $\psi:\B^k \to
  \Z$ as a multilinear polynomial $\widetilde{\psi}:\F_p^k \to \F_p$
  (recall that $\psi(x) < B < p$ for every $x \in \B^k$). Let
  $\widetilde{\psi}_q:\F_q^k \to \F_q$ be the polynomial function
  defined by the formal polynomial (i.e., a summation of monomials
  multiplied by coefficients) $\widetilde\psi$, where we take each
  coefficient of $\widetilde\psi$ modulo $q$. Since $\F_q$ is a
  subfield of $\H_q$, we can also view $\widetilde{\psi}_q$ as a
  multilinear polynomial from $\H_q^k$ to $\H_q$.

  Thus, we can express (\ref{eq:fq_representation}) as follows:
  \begin{align}
    \label{eq:hq_representation}
    \widetilde{\mathcal{P}}_q(\sigma) = \sum_{x \in
      \mathcal{D}_w(\H_q) } \sum_{y \in \mathcal{D}_s(\H_q) }
    \widetilde{\psi}_q \big(C_1\circ\tilde{\chi}^q(x,y), \ldots,
    C_k\circ\tilde{\chi}^q(x,y) \big)
  \end{align}
  (where the summation is over $\H_q$, which in this case is equal to
  summation over $\F_q$, hence the modulo $q$).\footnote{Since for
    every $x \in \mathcal{D}_w(\H_q)$ and $y \in \mathcal{D}_s(\H_q)$
    we have $\big(C_1\circ\tilde{\chi^q}(x,y), \ldots,
    C_k\circ\tilde{\chi^q}(x,y) \big) \in \B^k$, then each summand is
    in $\F_q$. Hence we can think of the summation as summation over
    $\F_q$.}

  For every $q \in Q$, we apply Lemma~\ref{lem:circuit_approx} with
  $\delta' = \frac{\delta}{2nc\log p}$, and $q$ as the prime
  number. We get a set of polynomials
  \begin{align*}
    \left\{ p_t: \F_q^m \to \F_q \right\}_{t\in[k]}
  \end{align*}
  (for every $q\in Q$), of degree $O \left( q \log \frac{k n \log
      p}{\delta} \right)$ each, such that for every $x\in \B^m$,
  \begin{align}
    \label{eq:pt_deg}
    \Pr \left[ \forall t\in[k] \quad p_t(x) = C_t(x) \right] \geq 1 -
    \frac{\delta}{2nc \log p}
  \end{align}
  (where the probability is taken over the random coin flips performed
  during the construction of $\{p_t\}_{t\in[k]}$).
  
  Since $\F_q$ is a subfield of $\H_q$, we can view $p_t : \F_q^m \to
  \F_q$ as a polynomial $\widetilde{p_t} : \H_q^m \to \H_q$ (for every
  $t\in[k]$). Then, for every $x\in \F_q^m$ we have
  $\widetilde{p_t}(x) = p_t(x)$. Thus, we get the following set of
  polynomials:
  \begin{align*}
    \left\{\widetilde{p_t} : \H_q^m\to \H_q \right\}_{t\in[k]},
  \end{align*}
  where for every $t\in[k]$, the degree of $\widetilde{p_t}$ is $O
  \left( q \log \frac{k n \log p}{\delta} \right)$.

  Applying a union bound, and using (\ref{eq:pt_deg}) yields:
  \begin{align}
    \label{eq:poly_approx_representation}
    \Pr \left[ \widetilde{\mathcal{P}}_q(\sigma) = \sum_{x \in
        \mathcal{D}_w(\H_q)} \sum_{y \in \mathcal{D}_s(\H_q)}
      \widetilde\psi_q \big(\widetilde{p_1}\circ\tilde{\chi}^q(x,y),
      \ldots, \widetilde{p_k}\circ\tilde{\chi}^q(x,y) \big) \right]
    \geq 1-\frac{\delta}{2c\log p}
  \end{align}
  (where the probability is taken over the random coin flips performed
  during the construction of $\{p_t\}_{t\in[k]}$, and the summation is
  over $\H_q$).\footnote{Again, since for every $x \in
    \mathcal{D}_w(\H_q)$ and $y \in \mathcal{D}_s(\H_q)$ we have
    $\big(\widetilde{p_1}\circ\tilde{\chi}^q(x,y), \ldots,
    \widetilde{p_k}\circ\tilde{\chi}^q(x,y) \big) \in \B^k$, then each
    summand is in $\F_q$. Hence, the summation is modulo $q$.}

  Next, we define the polynomial $\omega_q : \H_q \to \H_q$ by
  \begin{align*}
    \omega_q(x) = \sum_{y \in \mathcal{D}_s(\H_q)} \widetilde\psi_q
    \big(\widetilde{p_1}\circ\tilde{\chi}^q(x,y), \ldots,
    \widetilde{p_k}\circ\tilde{\chi}^q(x,y) \big)
  \end{align*}
  (where the summation is over $\H_q$). Note that for every $t\in[k]$,
  the composition of $\widetilde{p_t}$ and $\tilde{\chi}^q$ is a
  polynomial of degree
  \begin{align*}
    O \left( w q \log \frac{k n \log p}{\delta} \right)
  \end{align*}
  in $x$ (the first variable). Hence, by the multilinearity of
  $\widetilde\psi$,
  \begin{align}
    \label{eq:star}
    \deg(\omega_q) = O \left( w k q \log \frac{k n \log p}{\delta}
    \right).
  \end{align}
  By (\ref{eq:poly_approx_representation}) we have,
  \begin{align}
    \label{eq:mod_representation}
    \Pr \left[ \widetilde{\mathcal{P}}_q(\sigma) = \sum_{x \in
        \mathcal{D}_w(\H_q)} \omega_q(x) \right] \geq
    1-\frac{\delta}{2c\log p}
  \end{align}
  (where the probability is taken over the random coin flips performed
  during the construction of $\{p_t\}_{t\in[k]}$, and the summation is
  over $\H_q$).

  Once we established the above, we can finally describe Merlin's
  proof stream. The proof stream $\varphi$ consists of all the proof
  polynomials $\{\omega_q \}_{q \in Q}$. We send each polynomial by
  its list of coefficients, thus we need at most
  \begin{align*}
    O \Bigg( |Q| w k \log(p) \log \left(\frac{k n \log p}{\delta}
    \right) \cdot \log(p\log p)\Bigg)
  \end{align*}
  bits in order to write down the proof stream. Since $|Q| < c\log p$,
  we conclude:
  \begin{claim}
    \label{clm:proof_size}
    the total size of Merlin's proof stream $\varphi$ is
    \begin{align*}
      O \Big(w k \cdot \polylog \left(m, n, \delta^{-1} \right) \Big).
    \end{align*}
  \end{claim}
  
  Observe that it is possible to reconstruct $\widetilde{P}(\sigma)$
  from the polynomials given in Merlin's proof. We formalize this
  claim as follows:
  \begin{claim}
    \label{clm:tilde_eval}
    Given the set of values $\{ \sum_{x\in \mathcal{D}_w(\H_q)}
    \omega_q(x) \}_{q\in Q}$, it is possible to compute
    $\widetilde{P}(\sigma)$ with probability $1 - \delta/2$ (over the
    random coin tosses that were performed during the construction of
    $\{ \omega_q \}_{q\in Q}$).
  \end{claim}

  \begin{proof}
    Note that for every $q \in Q$ we have
    \begin{align*}
      \Pr \left[ \sum_{x\in \mathcal{D}_w(\H_q)} \omega_q(x) =
        \widetilde{\mathcal P}(\sigma)\pmod{q} \right] \geq 1 -
      \frac{\delta}{2c\log p}.
    \end{align*}
    Hence,
    \begin{align*}
      \Pr \left[\forall q \in Q \quad \sum_{x\in \mathcal{D}_w(\H_q)}
        \omega_q(x) = \widetilde{\mathcal P}(\sigma) \pmod{q} \right]
      \geq 1 - \frac{\delta}{2}.
    \end{align*}
    By the Chinese remainder theorem, given $\{ \widetilde{\mathcal
      P}(\sigma) \pmod{q} \}_{q\in Q}$ we can calculate
    \begin{align*}
      \widetilde{\mathcal P}(\sigma) \pmod{ \prod_{q\in Q} q}.
    \end{align*}
    Since we've chosen $Q$ such that $\prod_{q \in Q} q > p$, the
    claim follows.
  \end{proof}
  
  Another important property of the polynomials $\{\omega_q\}_{q\in
    Q}$ in the proof stream, is that given a sequential, one-pass
  access to the input stream, it is possible to efficiently evaluate
  each polynomial at a specific point. Formally, we show:
  \begin{lemma}
    \label{lem:sub_alg}
    For every $q \in Q$, there exists a streaming algorithm
    $\mathcal{A}_q$ with an access to the common random string $r$,
    such that given a point in the finite field $\xi \in \H_q$, and a
    sequential, one-pass access to the input stream $\sigma$, the
    streaming algorithm $\mathcal{A}_q$ can evaluate $\omega_q(\xi)$
    using $O \big( s k \cdot \polylog(m, n, \delta^{-1}) \big)$ bits
    of space.
  \end{lemma}

  \begin{proof}
    First, recall that the descriptions of $\{ C_t\}_{t\in[k]}$ and
    $\psi$ are known to the verifier. Note that in order to compute
    $\omega_q(\xi)$ it is sufficient to compute and store the values
    of
    \begin{align*}
      \left\{\widetilde{p_t} \big( \tilde{\chi}^q_1(\xi,y), \ldots,
        \tilde{\chi}^q_m(\xi,y) \big) \right\}_{t\in[k], y\in
        \mathcal{D}_s(\H_q)},
    \end{align*}
    where $\{\widetilde{p_t} \}_{t\in[k]}$ are the approximation
    polynomials of the clauses $\{ C_t \}_{t\in[k]}$ over
    $\H_q$. Given these values we can compute
    \begin{align*}
      \left\{ \widetilde\psi_q \Big(\widetilde{p_1} \big(
        \tilde{\chi}^q_1(\xi,y), \ldots, \tilde{\chi}^q_m(\xi,y)
        \big), \ldots, \widetilde{p_k} \big( \tilde{\chi}^q_1(\xi,y),
        \ldots, \tilde{\chi}^q_m(\xi,y) \big) \Big) \right\}_{ y\in
        \mathcal{D}_s(\H_q)}
    \end{align*}
    monomial-by-monomial according to the description of $\psi$, and
    then compute $\omega_q(\xi)$ by summing term-by-term.
    
    Before we describe the algorithm, recall that during the
    construction of $\{ p_t \}_{t\in[k]}$ we defined an error
    correcting code $\ECC:\F_q^m \to \F_q^{100m}$ with relative
    distance $1/3$. Note that since $\ECC$ is a linear function, we
    can extend it (via the linear extension) to $\H_q$. We fixed
    \begin{align*}
      L = O \left( \log k + \log\frac{1}{\delta'} \right) = O \left(
        \log k + \log\frac{n \log p}{\delta} \right),
    \end{align*}
    and chose independently and uniformly $\iota_1, \ldots, \iota_L
    \in[100m]$, using the common random string $r$. Finally we
    approximated each of the $\vee$ gates by the following polynomial,
    \begin{align}
      \label{eq:eta}
      \eta( z_1, \ldots, z_m ) = 1 - \prod_{l=1}^{L} \Big(1 -
      \left(\ECC ( z_1, \ldots, z_m )_{\iota_l} \right)^{q-1} \Big).
    \end{align}

    Note that in order to compute
    \begin{align*}
      \widetilde{p_t} \big( \tilde{\chi}^q_1(\xi,y), \ldots,
      \tilde{\chi}^1_m(\xi,y) \big)
    \end{align*}
    for all $t\in[k]$ and $y\in \mathcal{D}_s(\H_q)$, it is sufficient
    to compute
    \begin{align*}
      \ECC ( \ell^t_1(\xi,y), \ldots, \ell^t_m(\xi,y))_{\iota_l}
    \end{align*}
    (where for every $i \in [m]$ and $t\in[k]$ the value
    $\ell^t_i(\xi,y)$ is either $\tilde{\chi}^q_i(\xi,y)$, or
    $1-\tilde{\chi}^q_i(\xi,y)$, or one of the constants $\B$;
    depending on the clause $C_t$), for all $\iota_l \in \{\iota_1,
    \ldots, \iota_L \}$, $t\in[k]$, and $y\in
    \mathcal{D}_s(\H_q)$. Then we can compute $\widetilde{p_t} \big(
    \tilde{\chi}^q_1(\xi,y), \ldots, \tilde{\chi}^q_m(\xi,y) \big)$
    according to (\ref{eq:eta}).

    Since $\ECC$ is a linear error correcting code, we can compute
    each
    \begin{align*}
      \ECC ( \ell^t_1(\xi,y), \ldots, \ell^t_m(\xi,y))_{\iota_l}
    \end{align*}
    incrementally. That is, we read the data stream $\sigma$
    element-by-element. At each step, when the $i$'th element arrives
    ($i\in[m]$), for every $y \in \mathcal{D}_s(\H_q)$ we compute
    $\tilde{\chi}^q_i(\xi,y)$ according to (\ref{eq:lagrange}), and
    then $\ell^t_i(\xi,y)$ according to the description of $C_t$. By
    the linearity of $\ECC$ we can compute $\ECC ( \ell^t_1(\xi,y),
    \ldots, \ell^t_m(\xi,y))_{\iota_l}$ by incrementally adding each
    \begin{align*}
      \ECC (0,\ldots,0,\ell^t_i(\xi,y), 0, \ldots, 0)_{\iota_1}
    \end{align*}
    at the $i$'th step.

    Observe that during the run over $\sigma$, the entire computation
    is performed element-by-element, and that we used at most $O
    \left( | \mathcal{D}_s(\H_q) | \cdot k \cdot L \cdot \log p
    \right)$ bits of space. Thus the overall space complexity is
    \begin{align*}
      O \Big( s k \cdot \polylog \left(m, n, \delta^{-1}\right) \Big).
    \end{align*}
  \end{proof}

  The last lemma helps us to show that with high probability Merlin
  cannot cheat Arthur by using maliciously chosen proof
  polynomials. We show that by evaluating the actual proof polynomials
  at a randomly chosen point, Arthur can detect a false proof with
  high probability. Formally:
  \begin{lemma}
    \label{lem:verify}
    For every $q \in Q$, given a polynomial $\hat{\omega}_q : \H_q \to
    \H_q$ of degree at most $O \left( w k q \log \frac{k n \log
        p}{\delta} \right)$,\footnote{More precisely, the degree is
      exactly as in \ref{eq:star}.} if $\hat{\omega}_q \neq \omega_q$
    then:
    \begin{align*}
      \Pr[\hat{\omega}_q(\xi) = \omega_q(\xi)] \leq \frac{\delta}{2},
    \end{align*}
    where the probability is taken over uniformly choosing at random
    an element $\xi \in \H_q$.
  \end{lemma}

  \begin{proof}
    Let $\xi$ be an element uniformly chosen from $ \H_q$. By the
    Schwartz-Zippel Lemma, we have
    \begin{align*}
      \Pr[\hat{\omega}_q(\xi) = \omega_q(\xi)] \leq
      \frac{\max\set{\deg(\omega_q), \deg(\hat{\omega}_q)}}{|\H_q|}
      \leq \frac{\delta}{2},
    \end{align*}
    where in order to get the last inequality we fix $p$ to be a
    sufficiently large prime number, of order
    \begin{align*}
      \frac{1}{\delta} \cdot \poly(m,n).
    \end{align*}
  \end{proof}

  Finally, building upon the aforementioned lemmas, we can present the
  $\AM$ algorithm for the approximation of $\mathcal{P}(\sigma)$:

  \newpage
  \begin{algorithm}
    \caption{The Canonical $\AM$ streaming algorithm}
    \small \medskip \noindent

    {\bf The prover (Merlin):}
    \begin{enumerate}
    \item Choose $\iota_1,\ldots,\iota_L\in[100m]$ using the common
      random string $r$.
    \item Construct $\varphi$ that consists of all the proof
      polynomials $\{ \omega_q \}_{q \in Q}$.
    \item Send (via streaming) $\varphi = \{ \omega_q \}_{q \in Q}$ to
      the verifier.
    \end{enumerate}

    {\bf The verifier (Arthur):}
    \begin{enumerate}
    \item For every $q \in Q$, select uniformly at random $\xi_q \in
      \H_q$ (where the selection uses Arthur's private random coin
      tosses).
    \item Read Merlin's proof stream $\varphi = \{ \hat\omega_q \}_{q
        \in Q}$ and (incrementally) compute:
      \begin{enumerate}
      \item $\{ \hat{\omega}_q(\xi_q) \}_{q\in Q}$.
      \item $\left\{ \sum_{x\in \mathcal{D}_w (\H_q)}
          \hat{\omega}_q(x) \right\}_{q \in Q}$.
      \end{enumerate}
    \item Run $\{ \mathcal{A}_q \}_{q \in Q}$ in parallel, in order to
      compute $\{ \omega_q(\xi_q) \}_{q \in Q}$.
    \item If there exists $q\in Q$ for which $ \omega_q(\xi_q) \neq
      \hat{\omega}_q(\xi_q)$, return $\bot$.
    \item Otherwise, use $\left\{ \sum_{x\in \mathcal{D}_w (\H_q)}
        \hat{\omega}_q(x) \right\}_{q \in Q}$ to extract and return
      $\widetilde{\mathcal{P}}(\sigma)$.
    \end{enumerate}
  \end{algorithm}
  
  Last, we show that the aforementioned algorithm is an
  $\AM_{\epsilon,\delta}(S,W)$-streaming algorithm for
  $\mathcal{P}(\sigma)$, where
  \begin{itemize}
  \item $S = O \big( s k \cdot \polylog (m, n, \delta^{-1}) \big)$,
  \item $W = O \big( w k \cdot \polylog (m, n, \delta^{-1}) \big)$.
  \end{itemize}
  
  Indeed, given $\epsilon \ge 0$, $\delta>0$, a common random string
  $r$, and a data stream problem $\mathcal{P}$, our algorithm is a
  probabilistic data stream algorithm (denote it by $\mathcal{A}$),
  which has an oracle access to $r$. The algorithm simultaneously gets
  two streams: an input stream $\sigma$ and a proof stream $\varphi$,
  to both it has a sequential, one pass access. According to
  Claim~\ref{clm:proof_size}:
  \begin{align*}
    W = O \big( w k \cdot \polylog (m, n,\delta^{-1}) \big).
  \end{align*}
  
  As for the space complexity of $\mathcal{A}$, note that
  $\mathcal{A}$ stores $O(\log p)$ random values $\{\xi_q\}_{q\in Q}$
  of size $O(\log p)$ each, which takes $\polylog(m, n, \delta^{-1}))$
  bits of space. In addition it uses $\polylog(m, n, \delta^{-1}))$
  bits of space for computing
  \begin{enumerate}
  \item $\{ \hat{\omega}_q(\xi_q) \}_{q\in Q}$.
  \item $\left\{ \sum_{x\in \mathcal{D}_w (\H_q)} \hat{\omega}_q(x)
    \right\}_{q \in Q}$.
  \end{enumerate}
  Observe that these values can be computed incrementally using a
  sequential, one-pass access to $\varphi$, simply by evaluating the
  polynomials monomial-by-monomial. According to
  Lemma~\ref{lem:sub_alg}, each of the $O(\log p)$ algorithms $\{
  \mathcal{A}_q \}_{q \in Q}$ we run in parallel takes
  \begin{align*}
    O \big( s k \cdot \polylog(m, n, \delta^{-1}) \big)
  \end{align*}
  bits of space. Thus the total space complexity is $S = O \big( s k
  \cdot \polylog(m, n, \delta^{-1}) \big)$.

  Recall that the only time that the algorithm used the common random
  string $r$, is while building the approximation polynomial for the
  disjunction in each $\{ \omega_q \}_{q \in Q}$. Since we constructed
  $|Q|$ such polynomials, and by Lemma~\ref{lem:circuit_approx}, the
  total number of random bits we read from $r$ is $\polylog(m, n,
  \delta^{-1})$. Furthermore, $\mathcal{A}$ also uses only
  $\polylog(m, n, \delta^{-1})$ private random coin tosses, as the
  only randomness it needs is for the selection of random $\xi_q \in
  \H_q$ for every $q \in Q$. Thus, the total randomness complexity of
  the algorithm is $\polylog(m, n, \delta^{-1})$.
  
  We finish the proof by showing the correctness of the algorithm:
  \begin{enumerate}
  \item \textbf{Completeness:} Assuming Merlin is honest, i.e.,
    $\omega_q = \hat{\omega}_q$ for every $q\in Q$; then by
    Claim~\ref{clm:tilde_eval} we can calculate $\widetilde{\mathcal
      P}(\sigma)$ with probability $1 - \delta/2$ over the common
    random string $r$, and by (\ref{eq:bivariate_representation}) we
    have
    \begin{align*}
      (1-\epsilon) \mathcal{P}(\sigma) \leq \widetilde{\mathcal
        P}(\sigma) \leq (1+\epsilon) \mathcal{P}(\sigma).
    \end{align*}
    Hence:
    \begin{align*}
      \Pr \left[ \left|
          \frac{\mathcal{A}(\sigma,\varphi)}{\mathcal{P}(\sigma)} -1
        \right| \le \epsilon \right] \ge 1 - \frac{\delta}{2}
    \end{align*}
  \item \textbf{Soundness}: If Merlin is dishonest, i.e., there exists
    $q \in Q$ for which $\omega_q \neq \hat{\omega}_q$, then by
    Lemma~\ref{lem:verify},
    \begin{align*}
      \Pr[\mathcal{A}(\sigma,\varphi) \neq \bot] \le \frac{\delta}{2},
    \end{align*}
    where the probability is taken over the private random coin tosses
    that $\mathcal{A}$ performs.
  \end{enumerate}
\end{proof}

\section{The MA Communication Complexity of \emph{Gap Hamming
    Distance}}
\label{sec:ghd}
In this section we show that every $\MA$ communication complexity
protocol for the \emph{Gap Hamming Distance} problem ($\GHD$) that
communicates $T$ bits and uses a proof of length $W$, must satisfy $T
\cdot W=\Omega(n)$, and therefore $T+W=\Omega(\sqrt{n})$.

In Section~\ref{sec:distinctelements}, we will use the lower bound on
the $\MA$ communication complexity of $\GHD$ to show a lower bound on
the $\widehat{\MA}$ streaming complexity of the \emph{Distinct
  Elements} problem. We note that the lower bound on the $\MA$
communication complexity of $\GHD$ also implies a lower bound on the
$\widehat{\MA}$ streaming complexity of computing the empirical
entropy of a data stream (see \cite{BCM06} for a formal definition of
the \emph{Empirical Entropy} problem).

For completeness, we show an $\MA$ communication complexity protocol
for $\GHD$ that communicates $O(T \log n)$ bits and uses a proof of
length $O(W \log n)$, for every $T \cdot W \ge n$. Thus we have a
tight bound (up to logarithmic factors) of $T \cdot
W=\tilde\Omega(n).$

\subsection{Lower bound}
\label{sec:ghd_lb}
In order to prove our lower bound on the $\MA$ communication
complexity of \emph{Gap Hamming Distance}, we first show a lower bound
on the $\MA$ communication complexity of \emph{Gap Orthogonality}, a
problem wherein each party gets a vector in $\BM^n$ and needs to tell
whether the vectors are nearly orthogonal, or far from being
orthogonal. We then apply the reduction from the \emph{Gap
  Orthogonality} problem to the \emph{Gap Hamming Distance} problem
(following \cite{She11}), and obtain our lower bound.

Formally, the \emph{Gap Orthogonality} problem is defined as follows:
\begin{definition}
  Let $n$ be an integer, and let $\zeta_0,\zeta_1>0$.  The \emph{Gap
    Orthogonality} problem is the communication complexity problem of
  computing the partial Boolean function $\ORT_{n,\zeta_0,\zeta_1} :
  \BM^n \times \BM^n \to \B$ given by
  \begin{align*}
    \ORT_{n,\zeta_0,\zeta_1}(x,y) =
    \begin{cases}
      1 & if\quad |\langle x,y\rangle| < \zeta_1 \\
      0 & if\quad |\langle x,y\rangle| > \zeta_0 \\
    \end{cases}.
  \end{align*}
  Denote $\ORT = \ORT_{n,\frac{\sqrt{n}}{4},\frac{\sqrt{n}}{8}}$.
\end{definition}

We restate the following theorem from \cite{She11}, which given two
finite sets $X,Y$, guaranties that if the inner product of a random
vector from $X$ and a random vector from $Y$ is highly concentrated
around $0$, then $X\times Y$ must be a small rectangle.
\begin{theorem}
  \label{thm:sherstov}
  Let $\delta > 0$ be a sufficiently small constant, and let $X,Y
  \subseteq \BM^n$ be two sets, such that
  \begin{align*}
    \Pr \left[ |\langle x,y \rangle| > \frac{\sqrt{n}}{4}\right] <
    \delta
  \end{align*}
  (where the probability is taken over selecting independently and
  uniformly at random $x \in X$ and $y \in Y$), then
  \begin{align*}
    4^{-n} |X| |Y| = e^{-\Omega(n)}.
  \end{align*}
\end{theorem}

Denote the uniform distribution on $\BM^n \times \BM^n$ by $\mu$.  We
get the next immediate corollary of Theorem~\ref{thm:sherstov},
\begin{corollary}
  \label{cor:sherstov}
  There exists a (sufficiently small) constant $\delta > 0$ such that
  for every rectangle $R \subseteq \BM^{n} \times \BM^{n}$ with
  $\mu(R) > 2^{-\delta n}$ we have
  \begin{align*}
    \mu \left( R \cap \ORT^{-1}(0) \right) \ge \delta \mu(R).
  \end{align*}
\end{corollary}

\begin{proof}
  Assume by contradiction that there exists a rectangle $R \coloneqq X
  \times Y \subseteq \BM^{n} \times \BM^{n}$ with $\mu(R) > 2^{-\delta
    n}$ that satisfies
  \begin{align}
    \label{eq:measure}
    \mu \left( R \cap \ORT^{-1}(0) \right) < \delta \mu (R).
  \end{align}
  Observe that
  \begin{align*}
    \Pr \left[|\langle x,y\rangle| > \frac{\sqrt{n}}{4}\right] =
    \frac{\mu \left( R\cap \ORT^{-1}(0) \right)}{\mu(R)}
  \end{align*}
  (where the probability is taken over selecting independently and
  uniformly at random $x \in X$ and $y \in Y$). Hence we can
  write~(\ref{eq:measure}) as
  \begin{align}
    \label{eq:probability}
    \Pr \left[ |\langle x,y\rangle|> \frac{\sqrt{n}}{4} \right]
    <\delta.
  \end{align}
  Note that
  \begin{align*}
    \mu(R) = \frac{|X|}{2^n} \cdot \frac{|Y|}{2^n} = 4^{-n}|X||Y|.
  \end{align*}
  If we choose $\delta$ to be sufficiently small,
  then~(\ref{eq:probability}) guaranties the precondition of
  Theorem~\ref{thm:sherstov}, and we get that $\mu(R)=e^{-\Omega(n)}$,
  in contradiction to the assumption that $\mu(R)>2^{-\delta n}$.
\end{proof}

In particular, Corollary~\ref{cor:sherstov} implies that every
rectangle $R\subseteq\BM^{n}\times\BM^{n}$ satisfies
\begin{align}
  \label{eq:corruption}
  \mu\left(R\cap \ORT^{-1}(0)\right)\geq\delta \mu(R)-2^{-\delta n}.
\end{align}

Next, using well known techniques (cf. \cite{RS04}), we show a lower
bound on the $\MA$ communication complexity of $\ORT$, relying on
Corollary~\ref{cor:sherstov}. Formally, we prove
\begin{theorem}
  \label{thm:ort}
  Let $\epsilon$ be a positive constant such that
  $\epsilon<\frac{1}{2}$. For every $\MA_\epsilon(T,W)$ communication
  complexity protocol for $\ORT$ we have $T\cdot W=\Omega(n)$, hence
  $T+W=\Omega(\sqrt{n})$.
\end{theorem}

\begin{proof}
  Fix $n$. Denote $\mathcal{R} = \BM^n \times \BM^n$. Assume that
  there exists an $\MA_\epsilon(T,W)$ communication complexity
  protocol for $\ORT$; denote it by $\mathcal{P}$. By a simple
  amplification argument we get that there exists an
  $\MA_{\epsilon'}(k,W)$ communication complexity protocol for $\ORT$,
  where $k=O(T\cdot W)$ and $\epsilon' = 2^{-CW}$ (for an arbitrary
  large constant $C$); denote it by $\mathcal{P'}$.

  Assume by contradiction that $k=o(n)$. We will show that our
  assumption that $k$ is asymptotically smaller than $n$ implies that
  the error probability of $\mathcal{P'}$ is greater than $2^{-CW}$,
  in contradiction.

  Denote Merlin's proof, a binary string of size at most $W$ bits, by
  $w$. Denote the random string that $\mathcal{P'}$ uses by $s$.
  Denote by $R_{s,w,h}\subseteq \mathcal{R}$ the set of all input
  pairs $(x,y)\in \mathcal{R}$ such that the history of $(x,y,s,w)$ is
  $h$.\footnote{For any input pair $(x,y) \in \mathcal{R}$ and any
    assignment $s$ to the random string of $\mathcal{P'}$ and any
    assignment $w$ to the proof supplied to the players, the string of
    communication bits exchanged by the two players on the inputs $(x,
    y)$, using the random string $s$ and the proof $w$, is called the
    history of $(x, y, s, w)$.} We state the following Lemma from
  \cite{RS04}:
  \begin{lemma}
    \label{lem:rs}
    For every $s,w,h$ we have $R_{s,w,h}=X_{s,w,h}\times Y_{s,w,h}$
    (where $X_{s,w,h}\subseteq \BM^n$ and $Y_{s,w,h}\subseteq \BM^n$),
    and for every $s,w$ the family $\{ R_{s,w,h} \}_{h\in\B^{k}}$ is a
    partition of $\mathcal{R}$.
  \end{lemma}
  Denote the answer that $\mathcal{P'}$ gives on $(x,y,s,w)$ by
  $\mathcal{P'}(x,y,s,w)$. Since the answer of $\mathcal{P'}$ on
  inputs in $R_{s,w,h}$ does not depend on $x$ and $y$, then for every
  input pair in $R_{s,w,h}$ the answer $\mathcal{P'}(x,y,s,w)$ is the
  same; denote it by $\mathcal{P'}(s,w,h)$. Next, define
  $H_{0}\subseteq \mathcal{R}$ to be the set of all input pairs
  $(x,y)\in \mathcal{R}$ such that
  \begin{align*}
    |\langle x,y\rangle| > \frac{\sqrt{n}}{4},
  \end{align*}
  and define $H_{1}\subseteq \mathcal{R}$ to be the set of all input
  pairs $(x,y)\in \mathcal{R}$ such that
  \begin{align*}
    |\langle x,y \rangle| < \frac{\sqrt{n}}{8}.
  \end{align*}
  Note that if we choose $x = (x_1, \ldots, x_n) \in \BM^n$ and $y =
  (y_1, \ldots, y_n) \in \BM^n$ independently and uniformly at random,
  then for every $i \in [n]$ the product $x_i\cdot y_i$ is also
  uniformly distributed. Thus, if we choose $z = (z_1,\ldots,z_n) \in
  \BM^n$ uniformly at random, then
  \begin{align}
    \label{eq:stirling}
    \mu(H_{1}) = \Pr_{(x,y)\in\mathcal{R}} \left[ |\langle x,y
      \rangle| < \frac{\sqrt{n}}{8} \right] = \Pr_{z\in
      \BM^n}\left[\left| \sum_{i=1}^n z_i \right| < \frac{\sqrt{n}}{8}
    \right] \ge c,
  \end{align}
  for some universal constant $c$.
  
  Next, for every rectangle $R\subseteq \mathcal{R}$, denote by
  $\alpha(R)$ the measure of $R$ in $\mathcal{R}$. Denote by
  $\beta_{0}(R)$ the measure of $R\cap H_{0}$ in $H_{0}$, and denote
  by $\beta_{1}(R)$ the measure of $R\cap H_{1}$ in $H_{1}$. Under
  these notations, we see that (\ref{eq:corruption}) implies that
  there exists a universal constant $\delta > 0$ such that for any
  rectangle $R\subseteq \mathcal{R}$ we have
  \begin{align*}
    \beta_{0}(R)\geq\delta \cdot\alpha(R) - 2^{-\delta n}.
  \end{align*}
  According to Equation~\ref{eq:stirling}, we know that $H_{1}$ is a
  set of probability at least $c$ in $\mathcal{R}$. Hence for every
  rectangle $R \subseteq \mathcal{R}$ we have $\beta_{1}(R) \leq
  \nicefrac{1}{c} \cdot \alpha(R)$. Therefore we have the following
  corollary,
  \begin{corollary}
    \label{cor:measure}
    There exist universal constants $\delta, \delta'>0$ such that
    every rectangle $R\subseteq \mathcal{R}$ satisfies
    \begin{align*}
      \beta_{0}(R)\geq\delta' \cdot\beta_{1}(R) - 2^{-\delta n}.
    \end{align*}
  \end{corollary}

  For any $s,w$, denote by $A_{0}(s,w) \subseteq \mathcal{R}$ the
  union of all sets $R_{s,w,h}$ such that $\mathcal{P'}(s,w,h) = 0$,
  and denote by $A_{1}(s,w) \subseteq \mathcal{R}$ the union of all
  sets $R_{s,w,h}$ such that $\mathcal{P'}(s,w,h) = 1$. Observe that
  $A_{0}(s,w)$ and $A_{1}(s,w)$ are disjoint, and that $A_{0}(s,w)
  \cup A_{1}(s,w) = \mathcal{R}$.

  Since each of $A_{0}(s,w)$ and $A_{1}(s,w)$ is a union of at most
  $2^{k}$ of the sets $X_{s,w,h} \times Y_{s,w,h}$, we see that
  Corollary~\ref{cor:measure} implies
  \begin{align}
    \label{eq:false_positive}
    \beta_{0}(A_{1}(s,w)) \ge \delta' \cdot\beta_{1}(A_{1}(s,w))-
    2^{k} \cdot 2^{-\delta n} \ge \delta'
    \cdot\beta_{1}(A_{1}(s,w))-o(2^{-W}).
  \end{align}

  Recall that $\beta_{1}( A_{1}(s,w) )$ is the fraction of inputs
  $(x,y)$ in $H_{1}$ such that $\mathcal{P'}(x,y,s,w) = 1$, and that
  $H_{1}$ is the set of ones of the problem. Thus for every input
  $(x,y)$ in $H_{1}$ there exists $w$ such that $(x,y)\in A_{1}(s,w)$
  with probability of at least $(1-\epsilon')$ over $s$. Since the
  number of possible proofs $w$ is at most $2^{W}$, by an averaging
  argument we get that there exists a proof that corresponds to at
  least $2^{-W}$ fraction of the inputs in $H_{1}$. Formally speaking,
  there exists at least one binary string $w$ of size at most $W$, and
  a set $H_1' \subseteq H_1$ that satisfies $\beta_1(H_1')\geq
  2^{-W}$, such that that for every $(x,y)\in H_1'$,
  \begin{align*}
    \Pr_{s}[(x,y)\in A_1(s,w)] \ge 1-\epsilon'.
  \end{align*}
  Therefore, there exists a constant $c_0$, such that with constant
  probability (over the random string $s$),
  \begin{align*}
    \beta_{1}(A_{1}(s,w)) > 2^{-(W+c_0)}.
  \end{align*}
  Hence, by (\ref{eq:false_positive}), with constant probability (over
  the random string $s$),
  \begin{align*}
    \beta_{0}( A_{1}(s,w) ) \ge \delta' \cdot 2^{-W-c_0} -o(2^{-W})
    \ge c_1 \delta' \cdot 2^{-W}.
  \end{align*}
  for some constant $c_1$. However, recall that
  $\beta_{0}(A_{1}(s,w))$ is the fraction of inputs $(x,y)$ in $H_{0}$
  for which $\mathcal{P'}(x,y,s,w) $ returns $1$. Thus there exists a
  constant $c_2$ such that,
  \begin{align*}
    \Pr \left[ \mathcal{P'}(x,y,s,w) = 1 \right] \geq c_2 \delta'
    \cdot 2^{-W}
  \end{align*}
  (where the probability is taken over both the random string $s$, and
  the uniform selection of $(x,y)\in H_{0}$). But $H_{0}$ is the set
  of zeros of the problem, so for every $(x,y)\in H_{0}$ the protocol
  answers $1$ with probability at most $\epsilon' \le 2^{-CW}$ (for an
  arbitrary large constant $C$), which is a contradiction.
\end{proof}

We established that for every $\MA_\epsilon(T,W)$ communication
complexity protocol for $\ORT$ we have $T\cdot W=\Omega(n)$. According
to the duplication argument in \cite{She11}, Theorem~\ref{thm:ort}
implies the following corollary for slightly different parameters of
the orthogonality problem.
\begin{corollary}
  \label{cor:ort}
  Let $\epsilon$ be a positive constant such that
  $\epsilon<\frac{1}{2}$. For every $\MA_\epsilon(T,W)$ communication
  complexity protocol for $\ORT_{n,2\sqrt{n},\sqrt{n}}(x,y)$ we have
  $T\cdot W=\Omega(n)$, hence $T+W=\Omega(\sqrt{n})$.
\end{corollary}

Next, we state the following reduction from \cite{She11} (repharsed):
\begin{lemma}
  \label{lem:sherstov_reduction}
  Let $n\in \N$ be a perfect square. For every input $x\in \BM^{n}$
  denote by $x^{m}$ ($m\in \N$) the string of length $n \cdot m$ that
  is composed of $x$ concatenated to itself $m-1$ times. Then, for
  every $(x,y)\in \ORT_{n,2\sqrt{n},\sqrt{n}}^{-1}(0) \cup
  \ORT_{n,2\sqrt{n},\sqrt{n}}^{-1}(1)$ we have
  \begin{eqnarray*}
    \ORT_{n,2\sqrt{n},\sqrt{n}}(x,y) & = & \neg\GHD_{10n+15\sqrt n,
      \sqrt n, \sqrt n}\left(x^{10}(-1)^{15\sqrt n},y^{10}(+1)^{15\sqrt n}\right) \\
    & & \wedge\: \GHD_{10n+15\sqrt n, \sqrt n, \sqrt n}\left(x^{10}(+1)^{15\sqrt n},y^{10}(+1)^{15\sqrt n}\right).
  \end{eqnarray*}
\end{lemma}

Note that due to the symmetry of the \emph{gap Hamming distance}
problem, a protocol for $\GHD_{10n+15\sqrt n, \sqrt n, \sqrt n}$
implies a protocol for $\neg\GHD_{10n+15\sqrt n, \sqrt n, \sqrt
  n}$. Hence, if we assume by contradiction that there exists an
$\MA_\epsilon(T,W)$ communication complexity protocol for
$\GHD_{10n+15\sqrt n, \sqrt n, \sqrt n}$, where $0<\epsilon<\frac14$
and $T \cdot W=o(n)$ (which in turn implies that there exists an
$\MA_\epsilon(T,W)$ communication complexity protocol for
$\neg\GHD_{10n+15\sqrt n, \sqrt n, \sqrt n}$, where
$0<\epsilon<\frac14$ and $T \cdot W=o(n)$), then by applying
Lemma~\ref{lem:sherstov_reduction} we get an $\MA_{2\epsilon}(T,W)$
communication complexity protocol for
$\ORT_{n,2\sqrt{n},\sqrt{n}}(x,y)$ such that $T\cdot W=o(n)$, in
contradiction to Corollary~\ref{cor:ort}. Thus we get the following
corollary,
\begin{corollary}
  Let $\epsilon$ be a positive constant, such that
  $\epsilon<\frac{1}{4}$. For every $\MA_\epsilon(T,W)$ communication
  complexity protocol for $\GHD_{10n+15\sqrt n, \sqrt n, \sqrt n}$ we
  have $T\cdot W=\Omega(n)$, hence $T+W=\Omega(\sqrt{n}).$
\end{corollary}

Finally, we note that in previous work \cite{CR11} provided a toolkit
of simple reductions that can be used to generalize a lower bound on
the communication complexity of \emph{gap Hamming distance} for every
reasonable parameter settings. Specifically, a lower bound for
$\GHD_{10n+15\sqrt n, \sqrt n, \sqrt n}$ implies a lower bound for
$\GHD = \GHD_{n, \sqrt n, \sqrt n}$. Moreover, we note that their
reduction is directly robust for $\MA$ communication complexity; thus
we conclude,
\begin{theorem}
  \label{thm:ghd_lb}
  Let $\epsilon$ be a positive constant, such that
  $\epsilon<\frac{1}{4}$. For every $\MA_\epsilon(T,W)$ communication
  complexity protocol for $\GHD$ we have $T\cdot W=\Omega(n)$, hence
  $T+W=\Omega(\sqrt{n})$.
\end{theorem}

\subsection{Upper bound}
In their seminal paper, Aaronson and Widgerson \cite{AW09} showed an
$\MA$ communication complexity protocol for the disjointness problem,
wherein the communication complexity is $O(\sqrt n \log n)$, and the
size of the proof is also $O(\sqrt n \log n)$.

We modify their protocol in order to show an $\MA$ communication
complexity protocol for $\GHD$, wherein the communication complexity
is $O(T \log n)$, and the size of the proof is $O(W \log n)$, for
every $T \cdot W \ge n$.

\begin{theorem}
  Let $T,W\in\N$ such that $T \cdot W \ge n$. Then, there exists an
  explicit $\MA_{1/3} (T \log n, W \log n)$ communication complexity
  protocol for $\GHD$.
\end{theorem}
\begin{proof}
  Let $T,W \in \N$ such that $T \cdot W \ge n$. Assume for simplicity
  and without loss of generality that $T \cdot W = n$ exactly. Let $a
  \coloneqq (a_1, \ldots, a_n) \in \BM^n$ be the input of Alice, and
  $b \coloneqq (b_1, \ldots, b_n) \in \BM^n$ be the input of Bob. Let
  each player define a bivariate function that represents its input;
  more precisely, let Alice define $f_a:[W] \times [T] \to \BM$ by
  \begin{align*}
    f_a(x,y)=a_{(x-1)T +y},
  \end{align*}
  and similarly, let Bob define $f_b:[W] \times [T] \to \BM$ by
  \begin{align*}
    f_b(x,y)=b_{(x-1)T +y}.
  \end{align*}

  Fix a prime $q\in[6n,12n]$. Note that $f_a$ and $f_b$ have unique
  extensions $\tilde f_a:\F_q^2\to \F_q$ and $\tilde f_b:\F_q^2\to
  \F_q$ (respectively) as polynomials of degree $(W-1)$ in the first
  variable, and degree $(T-1)$ in the second variable. Next, define
  the polynomial $s:\F_q\to \F_q$ by
  \begin{align*}
    s(x)=\sum_{y\in[T]} \tilde f_a(x,y)\tilde f_b(x,y).
  \end{align*}
  Note that the degree of $s$ is at most $2(W-1)$. Denote the Hamming
  distance of $a$ and $b$ by $\mathsf{HD}(a,b)$. Then,
  \begin{align}
    \label{eq:HD}
    \mathsf{HD}(a,b) = \frac{n - \sum_{x\in[W]} s(x)}{2}.
  \end{align}
  Thus, it is sufficient for one of the players to know $s$ in order
  to compute the Hamming distance.  We define the following $\MA$
  communication complexity protocol:

  \begin{algorithm}
    \caption{$\MA$ Communication Complexity Protocol for $\GHD$}
    \small \medskip \noindent
    \begin{enumerate}
    \item Merlin sends Alice a message that consists of the
      coefficients of a polynomial $s':\F_q\to \F_q$ of degree at most
      $2(W-1)$, for which Merlin claims that $s'=s$.
    \item Bob uniformly picks $r \in\F_q$, and sends Alice a message
      that consists of $r$ and
      \begin{align*}
        \tilde f_b(r,1),\ldots,\tilde f_b(r,T).
      \end{align*}
    \item Alice computes $s(r)=\sum_{y\in[T]} \tilde f_a(r,y)\tilde
      f_b(r,y)$ and $s'(r)$. If $s(r)=s'(r),$ Alice computes
      $\mathsf{HD}(a,b) = \frac{n - \sum_{x\in[W]} s'(x)}{2}$ and
      returns the result. Otherwise, Alice rejects the proof and
      returns $\bot$.
    \end{enumerate}
  \end{algorithm}

  Note that Merlin sends the coefficients of a polynomial of degree at
  most $2(W-1)$ over a finite field of cardinality $O(n)$. Hence the
  size of the proof is $O(W \log n)$. In addition, note that the
  entire communication between Alice and Bob consists of sending the
  element $r$ and the $T$ evaluations of $\tilde f_b$ (in step 2 of
  the algorithm). Hence the total communication complexity is $O(T
  \log n)$.

  If Merlin is honest, then Alice can directly compute
  $\mathsf{HD}(a,b)$ with probability $1$, as according to
  (\ref{eq:HD}) the Hamming distance of $a$ and $b$ can be inferred
  from $s$. Otherwise, if $s' \neq s$ then by the Schwartz-Zippel
  Lemma
  \begin{align*}
    \Pr [s(r)=s'(r)] \leq \frac{2(W-1)}{q}\leq\frac{1}{3},
  \end{align*}
  (where the probability is taken over the random selection of
  $r\in\F_q$).  Thus the test fails with probability at least $2/3$.
\end{proof}

\section{The AM Streaming Complexity of \emph{Distinct Elements}}
\label{sec:distinctelements}
In this section we show an application of the canonical $\AM$
streaming algorithm for the \emph{Distinct Elements} problem.  In the
regular data stream model (without any probabilistic proof system), it
is well known (cf. \cite{M05}) that the space complexity of the
\emph{Distinct Elements} problem is lower bounded by the size of the
alphabet of the data stream (for sufficiently long data streams). In
contrast, using the canonical $\AM$ streaming algorithm we show that
by allowing $\AM$ proofs, we can obtain a tradeoff between the space
complexity and the size of the proof.

Furthermore, we then rely on our lower bound on the $\MA$
communication complexity of the $\GHD$ problem, in order to show a
matching lower bound on the $\widehat{\MA}$ streaming complexity of
\emph{Distinct Elements}.

\subsection{Upper Bound}
We show that for every $s,w \in \N$ such that $s \cdot w \ge n$ (where
$n$ is the size of the alphabet) there exists an $\AM$ streaming
algorithm for the \emph{Distinct Elements} problem that uses a proof
of size $\tilde O (w)$ and a space complexity $\tilde O (s)$. For
example, by fixing $w = n$, we have an $\AM$ streaming algorithm for
the \emph{Distinct Elements} problem that uses only a polylogarithmic
(in the size of the alphabet and the length of the stream) number of
bits of space.

Formally, we show:
\begin{theorem}
  \label{thm:de}
  For every $s,w\in\N$ such that $s \cdot w \ge n$, there exists an
  explicit $\AM_{0,1/3} \big(s \cdot \polylog(m,n) ,\: w \cdot
  \polylog(m,n)\big)$ streaming algorithm for the \emph{Distinct
    Elements} problem, given a data stream $\sigma=(a_1,\ldots,a_m)$
  with alphabet $[n]$.
\end{theorem}

The idea behind the proof of Theorem~\ref{thm:de} is simply noting
that we can indicate whether an element $j$ appears in the stream, by
the disjunction of the element indicators of $j\in[n]$ in all of the
positions of the stream (i.e., $\chi_1(j), \ldots, \chi_m(j)$). Then
we can represent the number of distinct elements as a sum of
disjunctions, and use the canonical $\AM$ streaming algorithm in order
to solve the \emph{Distinct Elements} problem. Formally,

\begin{proof}
  Recall that the \emph{Distinct Elements} problem is the data stream
  problem of computing (exactly) the following function:
  \begin{align*}
    F_0(\sigma) = \left|\left\{ i\in[n] \; : \; \exists j\in[m] \;\;
        a_{j}=i \right\}\right|.
  \end{align*}
  Observe that for every data stream we can write $F_0(\sigma)$ as
  \begin{align*}
    \sum_{j=1}^n \big( \chi_1(j) \vee \chi_2(j) \vee \ldots \vee
    \chi_m(j) \big).
  \end{align*}

  Let $\sigma=(a_1,\ldots,a_m)$ be a data stream with alphabet $[n]$.
  Let $s,w\in\N$ such that $s \cdot w \ge n$, let $\epsilon = 0$, and
  let $\delta = 1/3$.  By Theorem~\ref{thm:main} we have an explicit
  $\AM_{\epsilon,\delta}(S,W)$-streaming algorithm for computing
  $F_0(\sigma)$, where $S = O \big( s \cdot \polylog(m, n) \big)$ and
  $W = O \big( w \cdot \polylog (m, n) \big)$.
\end{proof}

\subsection{Lower bound}
In the rest of this section we consider the $\widehat{\MA}$ model. As
we mentioned in Section~\ref{sec:pps} the $\widehat{\MA}$ model,
wherein the verifier has a free access to the proof, is stronger than
the $\MA$ model, wherein the proof is being streamed. Hence the lower
bound we prove holds for both models.

As implicitly shown in \cite{IW03}, the communication complexity
problem of $\GHD$ reduces to the data stream problem of \emph{Distinct
  Elements}. We note that the foregoing reduction can be adapted in
order to reduce the $\MA$ communication complexity problem of $\GHD$
to the $\widehat{\MA}$ problem of approximating the number of distinct
elements in a stream within a multiplicative factor of $1 \pm
1/\sqrt{n}$. Together with our lower bound on the $\MA$ communication
complexity of $\GHD$, this implies the following:

\begin{theorem}
  \label{thm:de_lb}
  Let $\delta < \frac{1}{4}$. For every
  $\widehat{\MA}_{\frac{1}{\sqrt{n}}, \delta} (S,W)$ streaming
  algorithm for approximating the number of distinct elements in a
  data stream $\sigma=(a_{1},\ldots, a_{m})$ (over alphabet $[n]$) we
  have $S\cdot W=\Omega(n)$, hence $S+W=\Omega(\sqrt{n})$.
\end{theorem}

\begin{proof}
  Let Alice hold a string $x\in\BM^n$ and Bob hold $y\in\BM^n$. Alice
  can convert her string $x=(x_1,\ldots,x_n)$ to a data stream over
  the alphabet $\Sigma = \big\{ (i,b) \;|\; i\in[n],\;b\in\BM \big\}$
  in the following manner:
  \begin{align*}
    \sigma_A=\big((1,x_1),(2,x_2),\ldots,(n,x_n)\big).
  \end{align*}
  Similarly, Bob can convert $y=(y_1,\ldots,y_n)$ to
  \begin{align*}
    \sigma_B=\big((1,y_1),(2,y_2),\ldots,(n,y_n)\big).
  \end{align*}
  Observe that all of the elements in $\sigma_A$ are distinct, and
  that all of the elements in $\sigma_B$ are also distinct. In
  addition, note that the only way in which an element can appear
  twice in the concatenation of the streams is if $x_i=y_i$ for some
  $i\in[n]$. In fact, if we denote the number of distinct elements in
  $\sigma_A\circ\sigma_B$ by $d$, and denote the Hamming distance of
  $x$ and $y$ by $\mathsf{HD}(x,y)$, then we have the following
  relation:
  \begin{align}
    \label{eq:dist}
    d = n + \mathsf{HD}(x,y).
  \end{align}

  Alice and Bob can simulate running an $\widehat{\MA}$ streaming
  algorithm on the concatenation of their inputs by using a one-way
  $\MA$ communication complexity protocol, such that the number of the
  bits that are being communicated during the execution of the
  protocol is exactly the same as the number of bits of space that are
  used by the simulated $\widehat{\MA}$ streaming algorithm. Details
  follow.

  Say we have an $\widehat{\MA}_{\frac{1}{\sqrt{n}}, \delta} (S,W)$
  streaming algorithm $\mathcal{A}$ for approximating the number of
  distinct elements in $\sigma_A\circ\sigma_B$. Alice can run
  $\mathcal{A}$ on $\sigma_A$, using a proof $w$ of size $W$. After
  the algorithm finished processing the last element of $\sigma_A$,
  Alice sends the current state of her memory (which consists of at
  most $S$ bits) to Bob. Next, Bob sets his memory to the state that
  Alice had sent, uses the proof $w$, and completes the run of
  $\mathcal{A}$ over $\sigma_B$. Note that the total communication
  during the execution of the aforementioned protocol is at most $S$
  bits, as the data stream algorithm uses at most $S$ bits of space
  during its execution.

  As a conclusion, if there exists such $\mathcal{A}$ then by the
  reduction above there exists an $\MA$ communication protocol that
  outputs a $1 \pm 1/\sqrt{n}$ multiplicative approximation of $d$. By
  (\ref{eq:dist}) we can compute $\mathsf{\widetilde{HD}}(x,y)$, such
  that
  \begin{align*}
    \mathsf{HD}(x,y) - \sqrt{n} - \frac{\mathsf{HD}(x,y)}{\sqrt n} <
    \mathsf{\widetilde{HD}}(x,y) < \mathsf{HD}(x,y) + \sqrt{n} +
    \frac{\mathsf{HD}(x,y)}{\sqrt n},
  \end{align*}
  or
  \begin{align*}
    \mathsf{HD}(x,y) - 2\sqrt{n} < \mathsf{\widetilde{HD}}(x,y) <
    \mathsf{HD}(x,y) + 2\sqrt{n}.
  \end{align*}
  Thus we can solve $\GHD_{n, 2\sqrt n, 2\sqrt n}$ while communicating
  at most $O(S)$ bits and using a proof of at most $O(W)$ bits. Hence,
  using the toolkit of reductions provided in \cite{CR11} (see
  Section~\ref{sec:ghd_lb}) we can solve $\GHD$, while communicating
  at most $O(S)$ bits and using a proof of at most $O(W)$ bits. Thus,
  by Theorem~\ref{thm:ghd_lb} we have $S \cdot W = \Omega(n)$, hence
  $S+W = \Omega(\sqrt n)$.
\end{proof}

Note that in particular, Theorem~\ref{thm:de_lb} implies a lower bound
(with the same parameters) on the $\widehat{\MA}$ streaming complexity
of computing the \emph{exact} number of distinct elements in a stream.

Last, we also note that by a straightforward adaptation of the
reduction from the communication complexity problem of $\GHD$ to the
data stream problem of \emph{Empirical Entropy} (see \cite{CCM07}),
our $\MA$ lower bound on $\GHD$ also implies an $\widehat{\MA}$ lower
bound on the \emph{Empirical Entropy} problem.

\bibliographystyle{alpha} \bibliography{amstreams}
\end{document}